\documentclass[11pt,a4paper]{article}

\usepackage{amsmath,amsthm,amssymb,amsfonts}
\usepackage{mathrsfs}
\usepackage{authblk}
\usepackage{setspace}
\usepackage{mathrsfs}
\usepackage{cases}
\usepackage{graphicx}
\usepackage{amsthm}
\usepackage[usenames]{xcolor}

\numberwithin{equation}{section}
\newcommand{\ii}{{\rm i}}
\newcommand{\ee}{{\rm e}}
\newcommand{\unit}{1\!\!1}

\newcommand{\BbbR}{\mathbb{R}}
\newcommand{\BbbZ}{\mathbb{Z}}
\newcommand{\BbbC}{\mathbb{C}}

\newcommand{\W}{\mathcal{W}}
\newcommand{\What}{\widehat{\mathcal{W}}}

\DeclareMathOperator{\Realpart}{Re}
\DeclareMathOperator{\Imagpart}{Im}

\DeclareMathOperator{\supp}{supp}

\newcommand{\Df}{\mathscr{D}}
\newcommand{\Sf}{\mathscr{S}}

\theoremstyle{plain}
\newtheorem{thm}{Theorem}[section]
\newtheorem{lem}[thm]{Lemma}
\newtheorem{prop}[thm]{Proposition}

\theoremstyle{definition}
\newtheorem{defn}[thm]{Definition}

\theoremstyle{remark}
\newtheorem*{rem}{Remark}

\begin{document}

\title{Waiting for Unruh}
\author{Christopher J. Fewster,$^{*1}$ 
Benito A. Ju\'arez-Aubry$^{\dagger 2}$ 
and 
Jorma Louko$^{\dagger 3}$}

\affil{$^*$Department of Mathematics, University of York,\\
Heslington, York YO10 5DD, UK\\
$^1$chris.fewster@york.ac.uk}

\affil{$^\dagger$School of Mathematical Sciences, University of Nottingham,\\
Nottingham NG7 2RD, UK\\
$^2$pmxbaju@nottingham.ac.uk, benito.juarez@correo.nucleares.unam.mx \\
$^3$jorma.louko@nottingham.ac.uk}

\date{May 2016; revised June 2016\\[1ex]
Published in Class.\ Quant.\ Grav.\  {\bf 33}, 165003 (2016)}

\maketitle

\begin{abstract}
How long does a uniformly accelerated observer need to 
interact with a quantum field in order to record thermality 
in the Unruh temperature? 
We address this question for a pointlike Unruh-DeWitt detector, 
coupled linearly to a real Klein-Gordon field of mass $m\ge0$ and 
treated within first order perturbation theory, 
in the limit of large detector energy gap~$E_{\text{gap}}$. 
We first show that when the interaction duration $\Delta T$ is fixed, 
thermality in the sense of detailed balance cannot hold 
as $E_{\text{gap}}\to\infty$, 
and this property generalises from the Unruh effect to any 
Kubo-Martin-Schwinger state satisfying certain technical conditions. 
We then specialise to a massless field in four spacetime dimensions 
and show that detailed balance does hold when $\Delta T$ 
grows as a power-law in $E_{\text{gap}}$ as $E_{\text{gap}}\to\infty$, provided 
the switch-on and switch-off intervals are stretched proportionally to 
$\Delta T$ and the switching function 
has sufficiently strong Fourier decay. 
By contrast, if $\Delta T$ grows by 
stretching a plateau in which the interaction remains at constant strength 
but keeping the duration of the
switch-on and switch-off intervals fixed, 
detailed balance at $E_{\text{gap}}\to\infty$ requires $\Delta T$ 
to grow faster than any polynomial in~$E_{\text{gap}}$, 
under mild technical conditions. 
These results also hold
for a static detector in a Minkowski heat bath. 

The results limit the utility of the large $E_{\text{gap}}$ 
regime as a probe of thermality in time-dependent versions 
of the Hawking and Unruh effects, such as 
an observer falling into a radiating black hole. They may also
have implications on the design of prospective 
experimental tests of the Unruh effect. 
\end{abstract}

\singlespacing

\section{Introduction\label{sec:intro}} 

In the Unruh and Hawking 
effects in quantum field theory, an 
observer responds to a pure state of the quantum field 
as if the field were in a mixed, thermal state. 
In the case of the Unruh
effect~\cite{Unruh:1976db}, 
a uniformly linearly accelerated observer in Minkowski spacetime 
measures a temperature that is proportional to the acceleration
when the field is in the Poincar\'e invariant 
Minkowski vacuum state. 
In the case of the Hawking effect~\cite{Hawking:1974sw}, 
a stationary observer far from a black hole 
records thermal radiation coming from the hole when 
the field is in the quantum state that evolved from an appropriate  
no-particle state before the black hole was formed. 
Reviews of these phenomena and their interrelations 
are given in~\cite{Birrell:1982ix,Wald:1995yp}. 

To characterise the thermality in the Hawking and Unruh 
effects in terms of local measurements, we may
consider an observer who carries
a spatially localised quantum system 
that is coupled to the field, 
a particle detector~\cite{Unruh:1976db,DeWitt:1979}, 
and we may ask the observer to count the 
excitation and de-excitation 
transitions in the detector. 
The detector's response is
considered thermal when the 
transition probabilities satisfy
the detailed balance form 
\cite{Takagi:1986kn,Fredenhagen:1986jg} of the 
Kubo-Martin-Schwinger (KMS) condition
\cite{Kubo:1957mj,Martin:1959jp,Haag:1967sg}: 
the ratio of the excitation and de-excitation probabilities 
of energy gap $E_{\text{gap}}$ is $\ee^{-E_{\text{gap}}/T}$, 
where the positive constant 
$T$ is interpreted as the temperature. 
In this setting, 
there is broad evidence that 
detailed balance 
in the Unruh and Hawking effects emerges 
for an asymptotically stationary observer 
when the interaction time is long and the 
switching effects are negligible but the 
back-reaction of the observer on the quantum 
field still remains small
\cite{Unruh:1976db,Birrell:1982ix,Wald:1995yp,Higuchi:1993cya,Sriramkumar:1994pb,DeBievre:2006px,Lin:2006jw,Satz:2006kb,Louko:2007mu,Dappiaggi:2009fx,Juarez-Aubry:2014jba}. 

The purpose of this paper is to obtain precise asymptotic 
results about the long interaction time limit under which a 
detector records thermality in the Unruh effect, 
and specifically to examine how this limit depends 
on the detector's energy gap~$E_{\text{gap}}$. In short, 
how long does a uniformly accelerated observer need to wait 
for Unruh thermality to kick in at prescribed~$E_{\text{gap}}$, 
especially when $E_{\text{gap}}$ is large? 

We focus on the Unruh (rather than Hawking) effect
because of technical simplicity~\cite{Bisognano:1976za,Crispino:2007eb}, 
but we note that the conceptual 
issues about the meaning of a ``particle'' 
in a detector-field interaction are present already in 
the Unruh effect~\cite{Unruh:1983ms,Buchholz:2014jta,Buchholz:2015fqa}. 
We employ a pointlike two-level  
Unruh-DeWitt detector that is linearly 
coupled to a 
scalar field~\cite{Unruh:1976db,DeWitt:1979}, 
a system that models the interaction between atoms and the electromagnetic 
field when angular momentum interchange 
is negligible~\cite{MartinMartinez:2012th,Alhambra:2013uja}. 
Crucially, we assume that the detector-field coupling is 
proportional to a switching function 
$\chi \in C_0^\infty(\mathbb{R})$, 
a real-valued smooth compactly supported 
function of the detector's proper time. 
Finally, we assume throughout that the magnitude of the 
detector-field coupling is so small that 
first-order perturbation theory in the coupling remains applicable. 

In summary, we combine in a novel way 
the following three pieces of input: 
(a) we ask how the observation time required to record detailed 
balance depends on the energy scale at which the measurements are made; 
(b) we formulate this question using switching functions 
that are smooth and compactly supported; 
(c) our analysis is mathematically rigorous, especially 
regarding the distributional singularity of the Wightman function. 

We begin by considering not just the Unruh effect but 
any quantum state and detector motion 
in which the field's Wightman function, pulled back to the detector's worldline, 
satisfies the KMS condition
\cite{Kubo:1957mj,Martin:1959jp,Haag:1967sg} 
with respect to translations in the detector's proper time
at some positive temperature. 
We consider two families of switching functions, 
each of which involves a scaling parameter~$\lambda$. 
In the first family, the overall 
duration of the interaction is scaled by the factor~$\lambda$, 
including the switch-on and switch-off parts of the interaction: 
we refer to this as the adiabatic scaling. 
In the second family, the interaction is constant 
over an interval whose duration is proportional to~$\lambda$, 
while the preceding switch-on interval and the subsequent 
switch-off interval have 
$\lambda$-independent durations: 
we refer to this as the plateau scaling. 
The long interaction limit is hence obtained as 
$\lambda\to\infty$ within each family. 
For fixed energy gap, we verify that the 
$\lambda\to\infty$ limit does 
recover detailed balance in the KMS temperature, 
for both families. 
This observation is fully within expectations, 
and it justifies 
the long time response formulas 
that are standard in the literature
\cite{Unruh:1976db,Birrell:1982ix,Wald:1995yp,DeWitt:1979}. 
However, we also show that when the switching function is fixed, 
the detailed balance condition fails to hold 
at large values of the detector's energy gap. 
Physically, when the total interaction time is fixed, 
the detector does not have enough time to thermalise up to 
arbitrarily high energies. 

We may hence ask how fast $\lambda$ needs to 
increase as a function of $E_{\text{gap}}$ if the detector is to record detailed 
balance in the limit of large $E_{\text{gap}}$. 
A~waiting time that grows no 
faster than polynomially in $E_{\text{gap}}$ 
would presumably be realisable in experimental situations, 
such as prospective experimental tests of the Unruh effect, 
whereas a waiting time that needs to grow exponentially 
in $E_{\text{gap}}$ would presumably render 
an experimental verification of detailed balance 
at high $E_{\text{gap}}$ impractical. 

Our main result is to answer this question for the Unruh effect 
in four-dimensional Minkowski spacetime with a massless scalar field. 
Specialising to four spacetime dimensions is motivated by prospective 
applicability to future experiments, while setting 
the field mass to zero has the consequence that the 
only dimensionful parameter in the problem is the detector's 
proper acceleration~$a$. The detector's response then becomes
identical to that of a static detector in a 
static Minkowski heat bath of 
temperature~$a/(2\pi)$~\cite{Takagi:1986kn},
and all of our results will apply also there. 

For the adiabatic scaling, we show that 
detailed balance at large $E_{\text{gap}}$ 
can be achieved by letting $\lambda$ grow as a power-law 
in~$E_{\text{gap}}$, provided the Fourier 
transform of the switching function has sufficiently strong 
falloff properties. For the plateau scaling, by contrast, 
we show that 
detailed balance at large $E_{\text{gap}}$ requires 
$\lambda$ to grow faster than any strictly increasing, 
differentiable and 
polynomially bounded function of~$E_{\text{gap}}$. 

We conclude that in order to achieve detailed 
balance in the Unruh effect within an interaction 
time that grows as a power-law of $E_{\text{gap}}$ 
when $E_{\text{gap}}$ is large, 
it is crucial to stretch not just the overall duration 
of the interaction but also the intervals 
in which the interaction is switched on and off. 

We begin by introducing relevant mathematical 
notation and analytical preliminaries in Section~\ref{sec:prelim}. 
The detector model is introduced in Section~\ref{sec:coupling}.
Section \ref{sec:statdet-KMS}
establishes the precise connection between the KMS condition and 
detailed balance in the long interaction limit at fixed energy gap, 
and shows that this connection cannot hold 
for fixed interaction duration in the 
limit of large energy gap. 
We also give a technical formulation of what it 
means for detailed balance to hold at large energy gap 
in interaction time that grows as a 
power-law in the energy gap. 
The main results about the Unruh effect in 
four spacetime dimensions are 
obtained in Section~\ref{sec:unruh-section}. 
Section \ref{sec:conclusions} 
presents a summary and brief concluding remarks. 
Proofs of a number of auxiliary technical results 
are deferred to three appendices.

\section{Conventions and analytical preliminaries\label{sec:prelim}}

Our metric signature is mostly plus, and we work in in 
units in which $\hbar = c =1$. 
Complex conjugation is denoted by an overline. 
$O(x)$ denotes a quantity such that $x^{-1}O(x)$ 
is bounded as $x\to0$, and $O^\infty(x)$ 
denotes a quantity that falls off faster than any 
positive power of $x$ as $x\to0$. 
$C_0^\infty(\BbbR)$ denotes the space of smooth complex-valued functions on $\BbbR$ with 
compact support, i.e., vanishing identically outside a bounded set. 

The Fourier transform
of sufficiently regular $f:\BbbR\to\BbbC$ is defined by 
\begin{align}
\widehat{f}(\omega)=\int_{-\infty}^\infty ds\, f(s) \, \ee^{-\ii\omega s}\,.
\end{align}
As is well known (see e.g., \cite{Hormander:1983,Reed-Simon-vol2}) 
the transform of a Schwartz test function  
$f\in\Sf(\BbbR)$ (i.e., a smooth function which, together with its derivatives, 
decays faster than any inverse power at infinity) is also of the Schwartz class, 
while the transform of a smooth compactly supported
function extends to a holomorphic function on the whole complex plane
which obeys Paley--Wiener growth estimates and in particular decays
faster than any inverse power at infinity on the real axis. 

It will be important in what follows to be more
quantitative about exactly how fast the transform of a smooth
compactly supported test function can decay. On one hand, 
it is known that there are nontrivial smooth compactly supported $g$ with transforms obeying bounds of the form $|\widehat{g}(\omega)|\le K \ee^{-\gamma |\omega|^{q}}$ for any $0<q<1$ and positive constants $K$ and $\gamma$ -- see \cite{ingham,Fewster:2015hga} for constructions and examples. 
However, the transform of a nontrivial smooth compactly
supported function cannot decay exponentially fast: for if
\begin{align}\label{eq:expdecay}
|\widehat{g}(\omega)|\le K e^{-\gamma|\omega|}
\end{align} 
for some positive
constants $\gamma$ and $K$, then the Fourier inversion formula permits
us to extend $g(s)$ to a holomorphic function in the strip $|\Imagpart s|<\gamma$
and the fact that $g$ is compactly supported on the real axis then
implies that $g$ vanishes identically. 

This line of thought can 
be developed further: describing an open subset $\Omega\subset\BbbR$ as 
\emph{modest\/} if every bounded locally integrable function supported in $\Omega$ has a Fourier transform that is holomorphic in an open strip containing the real axis,
one has the following.
\begin{lem} 
If $g$ is smooth and compactly supported, 
and \eqref{eq:expdecay} holds except on a modest set~$\Omega$, 
then $g$ vanishes identically.
\label{lem:g-S}
\end{lem}
\begin{proof}
Let $\varphi$ be the characteristic function of~$\Omega$. 
Then $F = \varphi \widehat{g}$ has an inverse Fourier transform 
that is holomorphic in some open strip $S_1$  
around the real axis. 
As $\bigl|\bigl(1-\varphi(\epsilon)\bigr)\widehat{g}(\epsilon)\bigr| 
\leq K \ee^{-\gamma |\epsilon|}$ 
for all $\epsilon \in \mathbb{R}$, 
$(1-\varphi)\widehat{g}$ has an inverse transform that is holomorphic in some open 
strip $S_2$ around the real axis. 
Thus, $g$ extends from the real axis to a holomorphic 
function in the strip $S_1 \cap S_2$. 
But since $g$ has compact support on the real axis, $g$ has to vanish everywhere.
\end{proof}

Turning to distributions, the Fourier transform of a 
tempered distribution $U\in\Sf'(\BbbR)$
is defined so that $\widehat{U}(f)=U(\widehat{f})$ 
for any test function $f\in\Sf(\BbbR)$ 
(note that $f$ is a frequency-domain function here),
which implies the Plancherel formula 
$\widehat{U}(\overline{\widehat{g}})=2\pi U(\overline{g})$ 
for (time-domain) test functions~$g$. 

\section{The field-detector model\label{sec:coupling}} 

In this section we specify the field-detector model 
which the rest of the paper will analyse. 
The main point is the expression in 
\eqref{Probability}
and 
\eqref{ResponseFn}
for the detector's transition probability 
in a first-order perturbation theory treatment. 

Our detector is a pointlike quantum system that 
moves in a Lorentz-signature spacetime 
on the worldline $\textsf{x}(\tau)$, 
a smooth timelike curve 
parametrised by its proper time~$\tau$. 
We take the detector to be a two-level system, 
with the Hilbert space $\mathcal{H}_\text{D} \simeq \BbbC^2$ spanned 
by the orthonormal basis $\lbrace |0 \rangle, |1 \rangle \rbrace$, 
such that 
$H_{\text{D}} |0\rangle = 0 |0\rangle$ and $H_{\text{D}} |1 \rangle =
E |1 \rangle$, where $H_{\text{D}}$ is the detector's Hamiltonian 
with respect to $\tau$ and the constant $E\in\BbbR$ is the detector's energy gap. 
For $E > 0$ we may call $|0 \rangle$ the detector's ground state and
$|1 \rangle$ the excited state; for $E < 0$ the roles of $|0 \rangle$ and $|1 \rangle$ 
are reversed. 

The quantum field is a real scalar field~$\Phi$. 
When there is no coupling to the detector, 
we assume that the field is the free minimally coupled Klein--Gordon field of mass $m\ge 0$, 
and we assume the Hilbert space $\mathcal{H}_\Phi$ 
contains Hadamard state vectors \cite{Decanini:2005eg}
and admits a unitary time-evolution generated by a Hamiltonian.  
This is in particular the case when the spacetime is globally hyperbolic and stationary 
and the Hilbert space is the Fock space 
induced by the unique 
(sufficiently regular) ground or 
KMS states\footnote{The regularity required is minimal, 
but the states are then automatically quasifree and Hadamard.} 
in the $m>0$ case (see parts (i), (ii) and (iv) of
Theorems 5.1 and 6.2 in~\cite{Sanders:2013}). Of course
it also holds in the familiar setting of four-dimensional 
Minkowski spacetime for $m\ge 0$ in the vacuum representation.  

The Hilbert space of the total system is 
$\mathcal{H}_\Phi \otimes \mathcal{H}_\text{D}$, 
and the Hamiltonian of the total system is 
$H = H_{\Phi} \otimes \unit_{\text{D}} + \unit_{\Phi} \otimes
H_{\text{D}} + H_{\text{int}}$, where $H_{\text{D}}$ 
is the Hamiltonian of the detector, 
$H_{\Phi}$ is 
the Hamiltonian
of $\Phi$, and $H_{\text{int}}$ is the interaction Hamiltonian. 
We take $H_{\text{int}}$ to be 
\begin{align}
H_{\text{int}}(\tau) = c \chi(\tau) \Phi\bigl(\textsf{x}(\tau)\bigr) \otimes \mu(\tau)
\ , 
\label{Hint}
\end{align}
where $c \in \BbbR$ is a coupling constant, 
$\mu$ is the detector's monopole moment
operator and $\chi \in C_0^\infty(\mathbb{R})$ 
is a real-valued switching function that specifies how the 
interaction is turned on and off. 
We assume throughout that $\chi$ is nonvanishing somewhere. 

Suppose now that before the interaction starts, the detector is 
in the state $|0 \rangle$ and 
$\Phi$ is in some Hadamard state~$|\phi_i \rangle$. 
After the interaction has ceased, 
the probablity to find the detector in the state~$|1 \rangle$, 
regardless the final state of~$\Phi$, 
is given in first order perturbation theory by 
\cite{Unruh:1976db,Birrell:1982ix,Wald:1995yp,DeWitt:1979}
\begin{align}
\mathcal{P}
= c^2 \left| \langle 1 | \mu(0) | 0 \rangle \right|^2 \mathcal{F}(E)
\ , 
\label{Probability}
\end{align}
where 
\begin{equation}
\mathcal{F}(E) = \int_{-\infty}^\infty \! d\tau' \, 
\int_{-\infty}^\infty \! d\tau'' \, \chi(\tau') \chi(\tau'') \, 
\ee^{-\ii E (\tau'-\tau'')} 
\, \mathcal{W}(\tau',\tau'')
\ , 
\label{ResponseFn}
\end{equation}
and $\mathcal{W}\in\Df'(\BbbR\times\BbbR)$ is the pull-back of the Wightman function of $\Phi$ in the state $|\phi_i \rangle$ to the detector's worldline. 
$\mathcal{F}$~is called the response function of the detector, 
or the (smeared) power spectrum of the field's vacuum noise. 
As $|\phi_i \rangle$ is by assumption Hadamard, 
$\mathcal{W}$ is a well-defined distribution 
\cite{Hormander:1983,Fewster:1999gj} 
and $\mathcal{F}(E)$ is hence well 
defined pointwise for each~$E$. 
Note that $\mathcal{F}$ is real-valued since 
$\overline{\mathcal{W}(\tau',\tau'')} = \mathcal{W}(\tau'',\tau')$.
Equation \eqref{ResponseFn} is
more formally written as $\mathcal{F}(E)=\mathcal{W}(\overline{\chi_E},\chi_E)$,
where $\chi_E(\tau')=\ee^{\ii E\tau'}\chi(\tau')$. 
In this form it is clear that $\mathcal{F}$ is non-negative because of the positivity property 
$\mathcal{W}(\overline{f},f)\ge 0$ for all test functions~$f$, which ultimately arises from the positivity condition on states; here we also use the fact that $\chi$ is real-valued. 

In summary, the response function 
$\mathcal{F}$ \eqref{ResponseFn}
encodes how the transition probability depends on the 
field's initial state, the detector's energy gap, 
the detector's trajectory and the detector's switch-on and swich-off. 
The internal structure of the detector and the overall coupling strength 
enter only via the constant overall factor in~\eqref{Probability}. 
We shall from now on 
refer to the response function as the probability.

\section{Stationary detector in a KMS state\label{sec:statdet-KMS}} 

In this section we first give two mathematically 
precise formulations 
of the observation that the transition rate of a 
stationary detector in a KMS state 
satisfies the detailed balance condition provided 
the detector operates for a long time 
and the switching effects are negligible.  
We then show that the approach to detailed balance
in the long time limit cannot be uniform when 
the detector's energy gap increases, 
under a set of mild technical conditions on the switching. 
Finally, we formulate mathematically the question of how 
long the interaction needs to last in order
for detailed balance to hold at a given 
energy gap when the energy gap is large.

\subsection{The KMS condition and detailed balance}

We consider a pulled-back two-point function $\mathcal{W}$ that  
depends on its two arguments only through their difference, 
\begin{align}
\mathcal{W}(\tau',\tau'') = \W (\tau'-\tau'')
\ , 
\label{eq:Wstat-def}
\end{align}
and $\W$ will from now on denote the distribution   
on the right-hand side of~\eqref{eq:Wstat-def}. 
The detector's response is then invariant under translations in~$\chi$,
and is now given by
\begin{align}
\mathcal{F}(E) = 
\W(\overline{\chi_E}\star \chi_E) =\frac{\widehat{W}(|\widehat{\chi_E}|^2)}{2\pi}
\ .
\end{align}
This situation arises whenever the detector's 
trajectory and the state of the field are stationary 
with respect to the same notion of time translations, 
such as when they are both invariant under a Killing vector 
that is timelike in a neighbourhood of the trajectory. 
We refer to such trajectories as stationary. 
The only time dependence in the detector's response 
comes then from the switching function~$\chi$. 

As $\W$ arises from a Wightman two-point function of a state,
the distribution $\W\in\Df'(\BbbR)$ is of positive type, in the sense that
$\W(\overline{f}\star \widetilde{f})\ge 0$ for all test functions $f\in C_0^\infty(\BbbR)$,
where $\star$ denotes convolution and $\widetilde{f}(s)=f(-s)$. By the Bochner--Schwartz
theorem~\cite{Reed-Simon-vol2}, it follows that $\W$ is in fact a tempered distribution, 
$\W\in\Sf'(\BbbR)$, and that its distributional Fourier transform $\widehat{\W}$ is a polynomially bounded (positive) measure. For our purposes, it will be convenient (though in most cases not absolutely necessary)
to assume that $\widehat{\W}$ is in fact a continuous and polynomially bounded {\em function},
which is necessarily non-negative. This holds, for example, if $\W$ can be written as the sum of an absolutely integrable function and a distribution with compact support (for the first has a continuous transform, vanishing at infinity, while the second has a transform that is analytic and polynomially bounded
on the real axis). 
The latter is a fairly mild condition: by general properties of Hadamard states, $\W$ can be
decomposed (nonuniquely) as a sum of a smooth function and a distribution of compact support, provided
that no distinct points on the detector's worldline can be joined by a null geodesic. The
condition that the smooth part falls off sufficiently to be integrable will be satisfied in the
concrete cases that we study. \emph{From now on, therefore, our standing assumption
is that $\W$ is a tempered distribution with a continuous, non-negative, polynomially bounded Fourier
transform.} In this case, the detector response can be written
\begin{align}
\mathcal{F}(E) = \frac{1}{2 \pi} \int_{-\infty}^\infty \! d \omega \, 
\left|\widehat{\chi}(\omega)\right|^2 \, 
\What(E+\omega)
\,. 
\label{Fsstat1}
\end{align}

The Kubo-Martin-Schwinger (KMS) 
definition of a thermal 
Wightman function \cite{Kubo:1957mj,Martin:1959jp,Haag:1967sg,Bros:1996mw}
is adapted to $\W$ by the following definition. 

\begin{defn}[\textbf{KMS condition}]
Suppose there is a positive constant $\beta$ 
and a holomorphic function $\W_{\BbbC}$ in the strip 
$S = \left\{ s \in \BbbC \mid -\beta < \Imagpart s < 0 \right\}$ such that 
\begin{enumerate}
\item[(i)]
$\W$ is the distributional boundary value
of $\W_{\BbbC}$ on the real axis;\footnote{That is, for each $\eta\in(0,\beta)$, 
there is a tempered distribution defined by $\W_\eta(f)=\int ds\, \W_{\BbbC}(s-i\eta)f(s)$
and $\W_\eta(f)\to \W(f)$ as $\eta\to 0+$ for all $f\in\Sf'(\BbbR)$.} 
\item[(ii)]
$\W_{\BbbC}$ has a distributional boundary value 
on the line $\Imagpart s = -\ii \beta$; 
\item[(iii)]
The two boundary values of $\W_{\BbbC}$ are linked by 
\begin{align}
\W_{\BbbC}(s-\ii 0) = \W_{\BbbC}(-s-\ii\beta+\ii 0)
\label{eq:KMS}
\end{align}
in `kernel notation' for $s\in\BbbR$;
\item[(iv)]
For any $0<a<b<\beta$ there is a polynomial $P$ so that $|\W(s)|<P(|\Realpart s|)$ for
all $s\in S$ with $-b\le \Imagpart s\le -a$.
\end{enumerate}
Then we say that $\W$ obeys the \emph{KMS condition\/} 
at temperature~$1/\beta$. 
\label{def:KMS}
\end{defn}

Definition \ref{def:KMS} is consistent with the KMS
property as defined for the full Wightman function
in~\cite{Kubo:1957mj,Martin:1959jp,Haag:1967sg,Bros:1996mw}. 
In particular, if a state in Minkowski spacetime
is KMS in the sense of~\cite{Kubo:1957mj,Martin:1959jp,Haag:1967sg,Bros:1996mw}, 
the $\W$ of a static detector is KMS in the sense of 
Definition~\ref{def:KMS}.

We recall next what is meant by detailed balance. 
\begin{defn}[\textbf{Detailed balance}]
A function $G: \BbbR \to \BbbR$ that satisfies 
\begin{align}
G(-\omega) = \ee^{\beta \omega} G(\omega)
\label{eq:detailed-balance}
\end{align}
for a constant $\beta>0$ is said to satisfy the detailed balance
condition at temperature~$1/\beta$. 
\end{defn}

The connection between KMS and detailed balance is
that if $\W$ is KMS, its Fourier transform satisfies detailed balance and vice versa. 
\begin{prop}
\label{prop:Wstat-detailed}
Under the standing assumptions stated above, 
$\W$
satisfies the KMS condition 
at temperature $1/\beta$ iff\/ $\What$ 
satisfies the detailed balance condition \eqref{eq:detailed-balance} 
at temperature~$1/\beta$. 
\end{prop}
\begin{proof}
The standing assumptions require that $\What(\omega)\le P(\omega)$ 
some polynomial $P$. If detailed balance holds then, in fact, $\What(\omega)\le \ee^{-\beta\omega}P(-\omega)$ for $\omega>0$ and  
\begin{align}
\W_{\BbbC}(s) \doteq  
\int_{-\infty}^\infty \frac{d \omega}{2\pi} \, \ee^{\ii \omega s} \,
\What(\omega)  
\end{align}
defines a holomorphic function in the strip $S$, obeying the growth estimates
required in the KMS condition. By 
direct calculation using detailed balance, one finds that 
$\W_{\BbbC}(s-\ii\epsilon)=\W_{\BbbC}(-s-\ii\beta+\ii\epsilon)$ for any $s\in\BbbR$ and $0<\epsilon<\beta$,
and so the KMS condition \eqref{eq:KMS} holds in the sense of distributional
boundary values on taking  $\epsilon\to 0+$. 

Conversely, suppose that the KMS condition holds. For any $f\in C_0^\infty(\BbbR)$,
and $0<\delta,\epsilon<\beta/2$ then 
$F(s)=\W_{\BbbC}(s) \widehat{f}(s+\ii\epsilon)$ 
is holomorphic in $S$
and $|F(s)|\to 0$ faster than any inverse power as $|s|\to\infty$ in $-\beta+\delta \le
\Imagpart s\le -\epsilon$. It is thus legitimate to deform an integration contour
to obtain
\begin{align}
\int_{-\infty}^\infty ds\,\W_{\BbbC}(s-\ii\epsilon)\widehat{f}(s) = 
\int_{-\infty}^\infty ds\,\W_{\BbbC}(s-\ii\beta+\ii\delta)\widehat{f}(s-\ii\beta+\ii\epsilon+\ii\delta)\,.
\end{align}
Taking the limit $\epsilon\to 0+$ and then $\delta\to 0+$, using the KMS condition,
gives $\widehat{\W}(f)=\W(\widehat{f})=\widetilde{\W}(\widehat{f_\beta})=\widehat{\W}(\widetilde{f_\beta})$
where $f_\beta(\omega)=\ee^{-\beta\omega}f(\omega)$. Detailed
balance follows because $f$ was arbitrary in $C_0^\infty(\BbbR)$ and $\widehat{\W}$
is continuous by the standing assumptions.
\end{proof}

\subsection{Long interaction limit}

If the interaction lasts for a long time and the 
switching effects remain small, 
\eqref{Fsstat1} suggests that $\mathcal{F}(E)$ 
should be approximately proportional to~$\What(E)$,
in which case 
Proposition 
\ref{prop:Wstat-detailed}
shows that the detector's response 
has approximately detailed balance 
if $\W$ obeys the KMS condition
\cite{Unruh:1976db,Birrell:1982ix,Wald:1995yp,DeWitt:1979,Takagi:1986kn}. 
We shall present two implementations of a long time limit 
in which these expectations do hold 
in the sense of a long time limit 
with fixed~$E$.

\subsubsection{Adiabatic scaling}

We first consider the one-parameter family of switching functions 
\begin{align} 
\chi_\lambda(\tau) \doteq \chi(\tau/\lambda)
\ , 
\label{eq:adiab-chi}
\end{align}
where $\chi$ is a fixed switching function and 
$\lambda$ is a positive parameter. 
We refer to this family as the adiabatically scaled switching. 
The effective duration of the interaction is $\lambda$ 
times the width of the support of~$\chi$, 
and the long time limit is $\lambda\to\infty$. 

As $\lambda\to\infty$, 
we may expect the response to diverge proportionally to $\lambda$ 
because of the growing duration of the interaction. 
A more useful quantity is hence 
$1/\lambda$ times the response function, which we denote by 
$\widetilde{\mathcal{F}}_\lambda$, and for which 
\eqref{Fsstat1} and \eqref{eq:adiab-chi} give the formula 
\begin{align}
\widetilde{\mathcal{F}}_\lambda(E)  
= \frac{1}{2 \pi} \int_{-\infty}^{\infty} \! d \omega \, 
\left|\widehat{\chi}(\omega)\right|^2 
\, 
\What(E + \omega/\lambda)
\ . 
\label{eq:Fglambda}
\end{align}
The long time limit is given by the following proposition. We denote by 
$||\cdot||$ the $L^2$ norm in $L^2(\BbbR, d\omega)$. 

\begin{prop}
\label{prop:pointwise-lambdainfty}
For each fixed $E\in\BbbR$, 
$\widetilde{\mathcal{F}}_\lambda(E) \to 
\widetilde{\mathcal{F}}_\infty(E)$ as $\lambda\to\infty$, 
where 
\begin{align}
\widetilde{\mathcal{F}}_\infty(E) \doteq 
\frac{1}{2 \pi} ||\widehat{\chi}||^2 \,  
\What(E)
\ . 
\label{eq:Fglambda-limit}
\end{align}
\end{prop}

\begin{proof}
Let $E\in\BbbR$ be fixed. 
By the polynomial bound on $\bigl|\What(\omega)\bigr|$, 
there then 
exist constants $A >0$, $B>0$ and $n \in \BbbZ_+$ 
such that for $\lambda\ge1$ we have 
$\bigl|\What(E + \omega/\lambda)\bigr| 
< A  + B (E + \omega/\lambda)^{2n} 
\le 
A + B (|E| + |\omega|/\lambda)^{2n} 
\le 
A + B (|E| + |\omega|)^{2n}$. As 
$\widehat\chi(\omega)$ decays at $\omega \to\pm \infty$ 
faster than any inverse power of~$\omega$, 
the integrand in \eqref{eq:Fglambda} 
is hence bounded for $\lambda\ge1$ by 
a $\lambda$-independent integrable function, 
and the 
$\lambda\to\infty$ limit in \eqref{eq:Fglambda} may be taken 
under the integral by dominated convergence. 
\end{proof}

\subsubsection{Plateau scaling}

We next consider a one-parameter family of switching functions 
that take a nonvanishing constant value over an interval of 
adjustable duration, but the switch-on interval and 
the switch-off interval have fixed duration. 
We refer to this family as the plateau switching. 

To construct the family, we choose a bump function 
$\psi \in C_0^\infty(\mathbb{R})$, a non-negative 
smooth function with the support $[0, \tau_s]$, 
where the positive constant $\tau_s$ will 
be the duration of the switch-on and switch-off periods. 
Let $\tau_p$ be another positive constant, and define 
\begin{align}
\chi_\lambda(\tau) = \int_{-\infty}^{\tau} d \tau' 
\left[ \psi(\tau') - \psi(\tau' - \tau_s - \lambda \tau_p) \right]
\ ,
\label{chiconstlambda}
\end{align}
where $\lambda$ is a positive parameter. 
As $\psi$ is smooth, 
it follows from \eqref{chiconstlambda} that 
$\chi_\lambda$ is smooth. As $\psi$ is non-negative and has support 
$[0, \tau_s]$, it follows from \eqref{chiconstlambda} that 
$\chi_\lambda$ has support 
$[0, 2 \tau_s + \lambda \tau_p]$, 
consisting of the switch-on interval $[0, \tau_s]$,
the plateau $[\tau_s,  \tau_s + \lambda \tau_p]$ 
of duration $\lambda\tau_p$, and the switch-off interval 
$[\tau_s + \lambda \tau_p, 2\tau_s + \lambda \tau_p]$. 
The long time limit is $\lambda\to\infty$. 
Note that the long time limit scales only the duration of the 
plateau but leaves the durations of
the switch-on and switch-off periods fixed. 

The useful quantity to consider is again 
$1/\lambda$ times the response function, which we now denote by 
$\breve{\mathcal{F}}_\lambda$, and for which 
\eqref{Fsstat1} and \eqref{chiconstlambda}
give the formula 
\begin{equation}
\breve{\mathcal{F}}_\lambda(E) 
= \frac{1}{\pi\lambda} \int_{-\infty}^\infty \! d \omega \,
\frac{1-\cos\bigl[(\lambda \tau_p +\tau_s) \omega \bigr]}{\omega^2}
\, \bigl|\widehat{\psi}(\omega)\bigr|^2 
\, 
\What(E + \omega)
\ ,
\label{Flambda2}
\end{equation}
using the relation 
$\widehat{f'}(\omega) = \ii\omega \widehat{f}(\omega)$ that holds for any 
$f\in C^\infty_0(\BbbR)$. 
The long time limit is given by the following proposition. 
\begin{prop}
\label{prop:pointwise-lambdainfty2}
For each fixed $E\in\BbbR$, 
$\breve{\mathcal{F}}_\lambda(E) \to 
\breve{\mathcal{F}}_\infty(E)$ as $\lambda\to\infty$, 
where 
\begin{align}
\breve{\mathcal{F}}_\infty(E) = 
\tau_p \bigl|\widehat{\psi}(0) \bigr|^2 \,\, \widehat{\W}(E)
\ . 
\label{eq:Fglambda-limit2}
\end{align}
\end{prop}

\begin{rem}
As $\widehat\psi(0)\ne0$ by the assumptions on~$\psi$, 
$\breve{\mathcal{F}}_\infty$ \eqref{eq:Fglambda-limit2} 
is not identically vanishing. 
\end{rem}

\begin{proof}
Let $E\in\BbbR$ be fixed. 
Writing 
$g_E(\omega) \doteq \bigl| \widehat{\psi}(\omega) \bigr|^2 \, 
\widehat{\W}(\omega + E)$
and changing the integration variable in 
\eqref{Flambda2}
 by $(\lambda \tau_p + \tau_s) \omega = u$, we have 
\begin{equation}
\breve{\mathcal{F}}_\lambda(E) = \frac{1}{\pi} \int_{-\infty}^\infty \! d u \, 
(\tau_p + \tau_s/\lambda) \left(\frac{1-\cos u}{u^2}\right) 
g_E \! \left( \frac{u}{\lambda \tau_p + \tau_s} \right)
\ .
\label{eq:Fglambda-limit2-aux1}
\end{equation}
Since $\psi \in C_0^\infty(\mathbb{R})$ and $\bigl|\What(\omega)\bigr|$ 
is polynomially bounded, 
$|g_E(\omega)|$ is bounded. 
For $\lambda\ge1$, the integrand in 
\eqref{eq:Fglambda-limit2-aux1} 
is hence bounded in absolute value by 
a multiple of the integrable function 
$(1-\cos u)/u^2$, and the limit $\lambda\to\infty$ 
may be taken under the integral by dominated convergence, 
yielding~\eqref{eq:Fglambda-limit2}. 
\end{proof}

\subsubsection{Detailed balance in the long interaction limit}

We may collect the observations in 
Propositions \ref{prop:Wstat-detailed}, 
\ref{prop:pointwise-lambdainfty}
and 
\ref{prop:pointwise-lambdainfty2}
into the following theorem, which establishes the equivalence 
of detailed balance and the KMS condition 
in the long interaction time limit within our switching functions. 

\begin{thm}
\label{theorem:KMS-lambdainfty}
Under the standing assumptions on~$\W$: 
\begin{enumerate}
\item[(i)]
For each fixed $E\in\BbbR$, 
$\widetilde{\mathcal{F}}_\infty(E) 
\doteq \lim_{\lambda\to\infty} \widetilde{\mathcal{F}}_\lambda(E)$ 
and $\breve{\mathcal{F}}_\infty(E) 
\doteq \lim_{\lambda\to\infty} \breve{\mathcal{F}}_\lambda(E)$ 
exist and are proportional to 
$\widehat{\W}(E)$ as given in 
\eqref{eq:Fglambda-limit}
and~\eqref{eq:Fglambda-limit2}; 
\item[(ii)]
$\widetilde{\mathcal{F}}_\infty$ 
and $\breve{\mathcal{F}}_\infty$ 
satisfy the detailed 
balance condition 
\eqref{eq:detailed-balance} 
iff\/ $\W$
satisfies the KMS condition of Definition~\ref{def:KMS}. 
\end{enumerate}
\end{thm}

\subsection{Non-uniformity in $E$ of the long interaction limit}

Theorem \ref{theorem:KMS-lambdainfty}
shows that in a KMS state the long time limits of the response functions 
$\widetilde{\mathcal{F}}_\lambda(E)$ and $\breve{\mathcal{F}}_\lambda(E)$
exist and satisfy the detailed balance condition~\eqref{eq:detailed-balance}. 
The sense of the long time limit in Theorem \ref{theorem:KMS-lambdainfty} 
is $\lambda\to\infty$ with fixed~$E$. 
We now show that the sense in which $\widetilde{\mathcal{F}}_\lambda(E)$ 
and $\breve{\mathcal{F}}_\lambda$ approach the detailed balance 
condition as $\lambda$ increases cannot be uniform in~$E$. 

It is helpful to rewrite the detailed balance condition 
\eqref{eq:detailed-balance} as 
\begin{align}
\beta = \frac{1}{\omega} \ln \! \left( \frac{G(-\omega)}{G(\omega)}\right)
\ ;
\label{DBC}
\end{align}
at $\omega=0$ \eqref{DBC} is understood 
in the limiting sense. 
This shows that if $G$ is positive and 
$G(\omega)$ is bounded above by a polynomial in~$\omega$, 
$G(\omega)$ can satisfy the detailed balance condition at $|\omega|\to\infty$ 
only if $G(\omega)$ is bounded by an exponentially 
decreasing function of~$\omega$ at $\omega\to\infty$. 
The following proposition shows that $\mathcal{F}$ 
\eqref{Fsstat1} cannot decrease this fast under mild technical 
conditions on~$\What$. 
As $\widetilde{\mathcal{F}}_\lambda$ \eqref{eq:Fglambda} and 
$\breve{\mathcal{F}}_\lambda$ \eqref{Flambda2} 
are obtained from \eqref{Fsstat1} with a specific choice for $\chi$, 
the same holds for $\widetilde{\mathcal{F}}_\lambda$ and 
$\breve{\mathcal{F}}_\lambda$ with with any fixed~$\lambda$. 

\begin{prop}
Suppose that, in addition to the standing assumptions on~$\W$, 
there exist positive constants  $b$ and $C$ such that 
$\What(\omega) \ge C$ for $-b\le\omega\le b$. Then 
$\mathcal{F}(E)$ \eqref{Fsstat1} cannot fall off 
exponentially fast as $E\to\infty$. 
\label{prop:key-falloff}
\end{prop}

\begin{proof}
We use proof by contradiction, 
showing that an exponential falloff of $\mathcal{F}(E)$ at $E\to\infty$ 
implies that $\chi$ is identically vanishing. 

As $\What(\omega)$ is positive and polynomially bounded from above,  
$\mathcal{F}(E)$ 
is well defined for all~$E$. By the bound on $\What(\omega)$ from below, we have 
\begin{align}
\mathcal{F}(E) &= \frac{1}{2 \pi} \int_{-\infty}^\infty \! d v \, 
\left|\widehat{\chi}(v-E) \right|^2 \, \What(v)
\notag
\\
& \ge
\frac{1}{2 \pi} \int_{-b}^b \! d v \, 
\left|\widehat{\chi}(v-E)\right|^2 \, \What(v) 
\notag
\\
& \ge 
\frac{C}{2 \pi} \int_{-b}^b \! d v \, 
\left|\widehat{\chi}(v-E)\right|^2
\ , 
\label{eq:est1}
\end{align}
where we have first changed the integration variable by 
$\omega = v -E$, then restricted the integration range to $-b \le v \le b$, 
and finally used the strictly positive bound below on~$\What(v)$. 

Suppose now that $E>0$, 
and suppose there are positive constants $A$ and $\gamma$ such that 
$\mathcal{F}(E) < A \, \ee^{-\gamma E}$. 
From \eqref{eq:est1} we then have 
\begin{align}
\int_{-b}^b \! d v \, 
\left|\widehat{\chi}(v-E)\right|^2 < A' \, \ee^{-\gamma E}
\ , 
\label{eq:est2}
\end{align}
where $A' = 2 \pi A C^{-1} >0$. 
Inserting the factor $\ee^{\gamma(E-v)/2}$
under the integral in \eqref{eq:est2} gives 
\begin{align}
\int_{-b}^b \! d v \, 
\ee^{\gamma(E-v)/2}
\left|\widehat{\chi}(v-E)\right|^2 < A'  \, \ee^{\gamma b/2} \, \ee^{-\gamma E/2}
\ . 
\label{eq:est3}
\end{align}
Setting in \eqref{eq:est3}
$E = (2k+1)b$, $k = 0,1,\ldots$, and summing over~$k$, 
we obtain 
\begin{align}
\int_{-\infty}^0 \! d \omega \, 
\ee^{-\gamma \omega/2}
\left|\widehat{\chi}(\omega)\right|^2 < \infty 
\ , 
\label{eq:est4}
\end{align}
and since $\left|\widehat{\chi}(\omega)\right|$ is even in~$\omega$, it follows that 
\begin{align}
\int_{-\infty}^\infty \! d \omega \, 
\ee^{\gamma |\omega|/2}
\left|\widehat{\chi}(\omega)\right|^2 < \infty 
\ . 
\label{eq:est5}
\end{align}
This implies that $\ee^{\gamma|\omega|/4}\widehat{\chi}(\omega)$ 
is in $L^2(\BbbR, d\omega)$. 

Let now 
\begin{subequations}
\label{eq:est6}
\begin{align}
\Xi(\omega) & \doteq \overline{\widehat{\chi}(\omega)}\left(1+ \ee^{\gamma|\omega|/4}\right)
\ , 
\\
\Phi_z(\omega) & \doteq \frac{\ee^{\ii z \omega}}{1+ \ee^{\gamma|\omega|/4}}
\ , 
\end{align}
\end{subequations}
where $- \gamma/4 < \Imagpart z  < \gamma/4$.
Both $\Xi$ and $\Phi_z$ are in $L^2(\BbbR, d\omega)$. We may hence 
define in the strip $- \gamma/4 < \Imagpart z  < \gamma/4$ the function 
\begin{align}
\tilde\chi(z) \doteq (2\pi)^{-1} (\Xi, \Phi_z)
\ , 
\label{eq:est7}
\end{align}
where $(\cdot , \cdot)$ denotes the inner product in $L^2(\BbbR, d\omega)$. 
A~straightforward computation shows that $\tilde\chi(z)$ is holomorphic in the strip. 

Formulas \eqref{eq:est6} and \eqref{eq:est7} show that 
$\tilde\chi(\tau) = \chi(\tau)$ for $\tau \in \BbbR$. 
Since $\chi$ is by assumption compactly supported on~$\BbbR$, 
the holomorphicity of $\tilde\chi$ implies that $\tilde\chi$ vanishes everywhere. 
As $\chi$ is by assumption nonvanishing somewhere, this provides the sought contradiction. 
\end{proof}

\subsection{Asymptotic thermality with power-law waiting time}

We have seen that while the response of a stationary detector in a 
KMS state satisfies detailed balance in the
long detection time limit, 
both with the adiabatic scaling and the 
plateau scaling of the detection time, 
detailed balance with any fixed detection time 
will not be satisfied when 
the detector's energy gap is sufficiently large. 
This leads to the main question of the paper: 
if we wish detailed balance 
to hold as the energy gap increases, 
how fast does the 
interaction time need to grow 
compared to the energy gap? 

This question is motivated by experimental considerations. 
If the required waiting time increases (say) exponentially in the energy gap, 
we may expect experimental testing of detailed balance 
to become rapidly impractical at large energies. 
A~waiting time that increases only as a power of the energy gap 
is however experimentally more promising even at large energies. 

We formalise a waiting time that increases as a power-law 
in the following definition. 

\begin{defn}
For a stationary detector, let $\lambda$ be a positive 
parameter that is proportional to the 
effective duration of the interaction time, and let 
$\mathcal{F}_\lambda$ be the corresponding response function. 
Suppose there exists a positive function $\lambda(E)$ 
such that $\lambda(E) \to \infty$ as $|E| \to \infty$, 
$\lambda(E)$ is bounded from above by a polynomial, 
and 
for all sufficiently large $|E|$ one has 
\begin{equation}
E \beta - \mathcal{B}^{-}\bigl(E,\lambda(E)\bigr) \leq  
\ln \! \left( \frac{\mathcal{F}_{\lambda(E)}(-E)}{\mathcal{F}_{\lambda(E)}(E)}\right) 
\leq E \beta + \mathcal{B}^{+}\bigl(E,\lambda(E)\bigr)
\ , 
\label{AsympDBC}
\end{equation}
where $\beta$ is a positive constant and 
$\mathcal{B}^{\pm}\bigl(E,\lambda(E)\bigr)/E \rightarrow 0$ as a power-law in $E$ 
as $|E| \rightarrow \infty$. 
We then say that the response is asymptotically thermal 
at temperature $1/\beta$ with power-law waiting time. 
\label{DefAsympKMS}
\end{defn}

In Section \ref{sec:unruh-section} 
we shall examine whether asymptotic thermality 
with power-law waiting time holds for the Unruh effect with our 
adiabatic scaling \eqref{eq:adiab-chi} 
and plateau scaling~\eqref{chiconstlambda}.

\section{Unruh effect\label{sec:unruh-section}} 

In this section we specialise to the Unruh effect 
for a massless scalar field in four spacetime dimensions.  
We first show that 
detailed balance in the long interaction time limit holds 
for fixed energy gap both with our 
adiabatic scaling \eqref{eq:adiab-chi} 
and plateau scaling~\eqref{chiconstlambda}, but 
detailed balance cannot hold for any fixed interaction time 
in the limit of large energy gap. 
We then show that asymptotic thermality with power-law waiting 
time holds for the adiabatic scaling but does not hold for the plateau 
scaling.

\subsection{Long interaction limit at fixed $E$}

We consider four-dimensional Minkowski spacetime, 
a massless scalar field in the 
Minkowski vacuum state, and a detector on the Rindler trajectory of 
uniform linear acceleration. In global Minkowski coordinates $(t,x,y,z)$, 
the trajectory reads 
\begin{align}
t &= a^{-1} \sinh(a \tau)
\ , 
\ \ 
x = a^{-1} \cosh(a \tau)
\ , 
\ \ 
y = y_0
\ , 
\ \ 
z = z_0
\ , 
\label{eq:rindler-trajectory}
\end{align}
where the positive constant $a$ 
is the proper acceleration and $y_0$ and $z_0$ are constants. 
We then have \cite{Unruh:1976db,Birrell:1982ix,Wald:1995yp,DeWitt:1979}
\begin{align}
\W(s) 
= \lim_{\epsilon\to0+}
\left(
- \frac{a^2}{16 \pi^2 \sinh^2\bigl(a(s-\ii\epsilon)/2\bigr)}
\right) 
\ , 
\label{eq:rindler-wightman}
\end{align}
where the limit indicates the
boundary value in the sense of Definition~\ref{def:KMS}. 
As is well known, 
$\W$ \eqref{eq:rindler-wightman} coincides with what is obtained by  
starting with the Wightman function of the quasifree thermal state of temperature 
$a/(2\pi)$ 
and pulling it back to a worldline that is  
static in the corresponding inertial frame~\cite{Takagi:1986kn}. 

$\W$ \eqref{eq:rindler-wightman} satisfies 
the KMS condition of Definition \ref{def:KMS} with $\beta= 2\pi/a$. 
The Fourier transform of $\W$ 
is given by \cite{Unruh:1976db,Birrell:1982ix,DeWitt:1979}
\begin{align}
\What(\omega) = \frac{\omega}{2\pi \bigl(\ee^{2\pi \omega /a} -1\bigr)}
\ , 
\label{eq:Unruh-hatWstat}
\end{align}
understood at $\omega=0$ in the limiting sense. 
It is clear from \eqref{eq:Unruh-hatWstat} that 
$\What$ satisfies the detailed balance 
condition \eqref{eq:detailed-balance} with $\beta= 2\pi/a$: 
this had to happen by Proposition \ref{prop:Wstat-detailed} since 
$\W$ satisfies the standing assumptions 
introduced in Section~\ref{sec:statdet-KMS}.

The response of a detector in the 
long interaction time limit with the 
adiabatic scaling \eqref{eq:adiab-chi} 
and the plateau scaling \eqref{chiconstlambda}
is obtained from 
Theorem~\ref{theorem:KMS-lambdainfty}.   
$\widetilde{\mathcal{F}}_\infty(E)$ and 
$\breve{\mathcal{F}}_\infty(E)$ are hence multiples of~$\What(E)$, 
and they satisfy the detailed balance condition
in the 
Unruh temperature $T_\text{U} = a/(2\pi)$. 
This is the Unruh effect~\cite{Unruh:1976db}. 

However, $\widetilde{\mathcal{F}}_\lambda(E)$ 
and $\breve{\mathcal{F}}_\lambda(E)$ 
cannot satisfy detailed balance
at large $|E|$ 
for any fixed~$\lambda$. 
This follows from Proposition 
\ref{prop:key-falloff} and the observation that 
${\mathcal{F}}(E)$ \eqref{Fsstat1}
is bounded from above by a polynomial in~$E$, 
which we verify in Appendix~\ref{app:polboundFUnruh}. 

We shall now examine whether 
$\widetilde{\mathcal{F}}_\lambda(E)$ 
and 
$\breve{\mathcal{F}}_\lambda(E)$
approach detailed balance at large $|E|$ 
in waiting time that grows as a power-law in~$E$, 
in the sense of Definition~\ref{DefAsympKMS}.

\subsection{Waiting for Unruh with adiabatic scaling}
\label{subsec:slowswitching}

For the switching functions \eqref{eq:adiab-chi} with adiabatic scaling, 
\eqref{eq:Fglambda} and \eqref{eq:Unruh-hatWstat} give 
\begin{align}
\widetilde{\mathcal{F}}_\lambda(E)
= \frac{1}{(2 \pi)^2} \int_{-\infty}^{\infty} \! d \omega \, 
\left|\widehat{\chi}(\omega)\right|^2 
\left( \frac{E+\omega/\lambda}{\ee^{2 \pi(E + \omega/\lambda)/a} - 1}\right)
\ , 
\label{eq:tildeFadiab}
\end{align}
while \eqref{eq:Fglambda-limit} and \eqref{eq:Unruh-hatWstat} give 
\begin{align}
\widetilde{\mathcal{F}}_\infty(E) = 
\frac{|| \widehat{\chi} ||^2}{{(2\pi)}^2 }
\frac{E}{\bigl(\ee^{2\pi E /a} -1\bigr)} 
\ . 
\label{eq:tildeFadiab-infty}
\end{align}
Can $\widetilde{\mathcal{F}}_\lambda(E)$ approach detailed balance 
at large $|E|$ in waiting time that grows as a power-law in~$E$? 

As a first step, we note that \eqref{eq:tildeFadiab-infty}
and Lemma \ref{Lem2} in Appendix \ref{app:auxiliary-main1} give the
inequalities 
\begin{equation}
\frac{E}{T_\text{U}} - \mathcal{B}(E,\lambda) 
\leq 
\ln \! \left(\frac{\widetilde{\mathcal{F}}_\lambda(-E)}{\widetilde{\mathcal{F}}_\lambda(E)} \right) 
\leq \frac{E}{T_\text{U}} + \mathcal{B}(-E,\lambda)
\ ,
\label{Bgen-used}
\end{equation}
where
\begin{align}
\mathcal{B}(E,\lambda) = 
\ln \! \left(1 + 
\frac{|| \omega \widehat{\chi} ||^2}{|| \widehat{\chi} ||^2}
\frac{ \left( \ee^{2\pi E/a} -1 \right)}{(6 a E/\pi) \lambda^2} \right)
\ . 
\label{Bgen}
\end{align}
These estimates are not sufficient to establish asymptotic 
thermality in power-law waiting time: 
from \eqref{Bgen} we see that a power-law $\lambda(E)$ 
implies $\mathcal{B}\bigl(E,\lambda(E)\bigr)/E \to 0$ as 
$E\to-\infty$, but achieving $\mathcal{B}\bigl(E,\lambda(E)\bigr)/E \to 0$ as 
$E\to\infty$ would require $\lambda(E)$ 
to grow exponentially in~$E$. 
However, in the following definition and theorem we show that 
there exists a class of switching 
functions for which the estimate in \eqref{Bgen} can be improved
and asymptotic thermality in power-law waiting time achieved. 

\begin{defn}
If $\psi \in C_0^\infty(\mathbb{R})$ 
satisfies 
\begin{align}
|\widehat{\psi}(\omega)| 
\leq C(B+|\omega|)^r 
\exp(-A |\omega|^q)
\label{eq:strong-fourier-decay-def}
\end{align}
for
some constant  
$q \in (0,1)$ and some 
some positive constants 
$A$, $B$ and~$C$, 
we say that 
$\psi$ has strong Fourier decay. 
\label{RapidFourier}
\end{defn}

\begin{rem}
For any $q\in(0,1)$, 
$C_0^\infty(\mathbb{R})$ contains non-negative 
functions that are not identically vanishing and 
have strong Fourier decay 
with \eqref{eq:strong-fourier-decay-def}~\cite{ingham,Fewster:2015hga}. 
\end{rem}

\begin{thm}
Suppose the switching function 
$\chi$ has strong Fourier decay, satisfying 
\begin{equation}
|\widehat{\chi}(\omega)| \leq C\kappa^{-1}(B+|\omega/\kappa|)^r \exp(-A |\omega/\kappa|^q)
\label{eq:Main1-inputest}
\end{equation}
for positive constants $A$, $B$, $C$, $r$ and $\kappa$ and 
a constant $q \in (0,1)$. 
Then $\widetilde{\mathcal{F}}_\lambda$ \eqref{eq:tildeFadiab} is asymptotically thermal 
with power-law waiting time, 
and we may choose $\lambda(E) = (2 \pi |E|/a)^{1+p}$ 
where $p > q^{-1} - 1$.
\label{thm:Main1}
\end{thm}

\begin{rem}
The constant $\kappa$ has been included in \eqref{eq:Main1-inputest} 
for dimensional convenience: 
$\kappa$~has the physical dimension of inverse length whereas 
the other constants are dimensionless. 
\end{rem}

\begin{proof}
By the remarks below~\eqref{Bgen}, 
it suffices to find 
$\mathcal{B}^{-}\bigl(E,\lambda(E)\bigr)$ for   
\eqref{AsympDBC} as $E\to\infty$. 

Let $\lambda(E)$ be as stated in the Theorem, let 
$E>0$, and let 
\begin{equation}
\mathcal{G}_\text{est}(E) \doteq 
\frac{2 a ||\widehat{\chi}||^2}{(2 \pi)^3} 
\left(\frac{2\pi E}{a}\right)^{-(1+p-q^{-1})/2} 
\ . 
\label{eq:G-est-def}
\end{equation}
We show in Lemma \ref{Lem-n1} in Appendix \ref{app:auxiliary-main1}
that
\begin{align}
\widetilde{\mathcal{F}}_{\lambda(E)}(E)  
\leq \widetilde{\mathcal{F}}_\infty(E) + \frac{2 \pi E}{a}\exp( -2\pi E /a) 
\mathcal{G}_\text{est}(E)
\label{eq:G-est-used}
\end{align}
holds for sufficiently large~$E$. 
Combining this with the inequality 
$\widetilde{\mathcal{F}}_\infty(-E) \leq \widetilde{\mathcal{F}}_{\lambda(E)}(-E)$
from Lemma~\ref{Lem2}, and using 
$\widetilde{\mathcal{F}}_\infty(E) = \ee^{E/T_{\text{U}}} \widetilde{\mathcal{F}}_\infty(-E)$, 
we find that for sufficiently large~$E$, 
\begin{align}
\ln \! 
\left(\frac{\widetilde{\mathcal{F}}_{\lambda(E)}(-E)}{\widetilde{\mathcal{F}}_{\lambda(E)}(E)} \right) 
\geq 
\frac{E}{T_{\text{U}}} - \mathcal{B}^{-}(E)
\ , 
\label{B-thm-ineq}
\end{align}
where 
\begin{align}
\mathcal{B}^{-}(E) & =  \ln \! 
\left(
1 + \frac{(2\pi)^3 \bigl(1 - \ee^{-2\pi E/a} \bigr) 
\mathcal{G}_\text{est}(E)}{a || \widehat\chi ||^2} 
\right)  
\nonumber 
\\[1ex]
& = 2 \left(\frac{2\pi E}{a} \right)^{-(1 + p - q^{-1})/2} + 
O \! \left( ( E/a)^{-(1 + p - q^{-1})}\right) 
\ . 
\label{B-thm}
\end{align}
Since $\mathcal{B}^{-}(E)/E$ has a power-law falloff at $E\to\infty$, 
\eqref{B-thm-ineq} provides 
the sought lower bound in~\eqref{AsympDBC}. 
\end{proof}

\subsection{Waiting for Unruh with plateau scaling}
\label{subsec:plateau}

For the switching functions \eqref{chiconstlambda} with plateau scaling, 
\eqref{Flambda2} and \eqref{eq:Unruh-hatWstat} give 
\begin{equation}
\breve{\mathcal{F}}_\lambda(E) 
= \frac{1}{2\pi^2\lambda} \int_{-\infty}^\infty \! d \omega \,
\frac{1-\cos\bigl[(\lambda \tau_p + \tau_s) \omega \bigr]}{\omega^2}
\, \bigl|\widehat{\psi}(\omega)\bigr|^2 
\left( \frac{E+\omega}{\ee^{2 \pi(E + \omega)/a} - 1}\right)
\ ,
\label{Flambda2Unruh}
\end{equation}
while \eqref{eq:Fglambda-limit2} and \eqref{eq:Unruh-hatWstat} give 
\begin{align}
\breve{\mathcal{F}}_\infty(E) = 
\frac{\tau_p}{2\pi} \bigl|\widehat{\psi}(0) \bigr|^2
\frac{E}{\bigl(\ee^{2\pi E /a} -1\bigr)}
\ . 
\label{FT2}
\end{align}
In the following theorem we provide a sense in which 
$\breve{\mathcal{F}}_\lambda(E)$ cannot approach 
detailed balance 
at large $|E|$ in waiting time that grows as a power-law in~$E$. 

\begin{thm}
Let $P: \mathbb{R}^+ \rightarrow \mathbb{R}^+$ 
be differentiable, strictly increasing and polynomially bounded,
and let $P(E) \to \infty$ as $E\to\infty$.  
Then $\breve{\mathcal{F}}_{P(|E|)}(E)$ is not asymptotically thermal 
with power-law waiting time in the sense of Definition~\ref{DefAsympKMS}. 
\label{thm:Main2}
\end{thm}

\begin{proof}
Adapting the proof of Lemma~\ref{Lem2}, we see that 
$\breve{\mathcal{F}}_{P(|E|)}(E) \ge \breve{\mathcal{F}}_\infty(E) > 0$ 
for all~$E$. We shall show that $E^{-1} \breve{\mathcal{F}}_{P(|E|)}(-E)$ 
is bounded as $E\to\infty$
but $E^{-1} \ee^{2\pi E/a} \breve{\mathcal{F}}_{P(|E|)}(E)$ 
is not bounded as $E\to\infty$. 
From this it follows that 
$\ee^{2\pi E/a} \breve{\mathcal{F}}_{P(|E|)}(E) / \breve{\mathcal{F}}_{P(|E|)}(-E)$ 
is not bounded as $E\to\infty$, and 
$\breve{\mathcal{F}}_{P(|E|)}(E)$ hence does not 
have the property of Definition~\ref{DefAsympKMS}. 

As preparation, we define the function $g \in C_0^\infty(\BbbR)$ by 
\begin{align}
g(\tau) \doteq \int_{-\infty}^\infty \! dt \, 
\psi( \tau - t) \psi(-t)
\ . 
\end{align}
By the convolution theorem, 
$\widehat{g}(\omega) = \widehat{\psi}(\omega) \widehat{\psi}(-\omega) 
= \bigl| 
\widehat{\psi}(\omega)
\bigr|^2$. It follows that 
$\widehat g$ is non-negative and even and has rapid decrease. 

Suppose now $E>0$. 
From \eqref{Flambda2Unruh} we have 
\begin{subequations}
\label{eq:plateau-prelim-both}
\begin{align}
\frac{\breve{\mathcal{F}}_{P(E)}(-E)}{2\pi E/a} 
& = 
\frac{2a}{(2\pi)^3} \left( \tau_p +  \frac{\tau_s}{P(E)} \right) 
H(2\pi E/a)
\ , 
\\[1ex]
\frac{\exp(2 \pi E/a)}{2\pi E/a} \breve{\mathcal{F}}_{P(E)}(E) 
& = 
\frac{2a}{(2\pi)^3} \left( \tau_p +  \frac{\tau_s}{P(E)} \right) 
I(2\pi E/a)
\ , 
\end{align}
\end{subequations}
where 
\begin{subequations}
\label{eq:plateau-both}
\begin{align}
H(\mathcal{E})
& \doteq \frac{1}{\mathcal{E}}
\int_{-\infty}^{\infty} 
\! du \, 
\frac{1- \cos u}{u^2} \, 
\widehat{g} \! \left(\frac{au}{2\pi\Lambda(\mathcal{E})}\right) 
\left( 
\frac{u/\bigl(\Lambda(\mathcal{E}) \bigr) 
- \mathcal{E}}{\ee^{u/(\Lambda(\mathcal{E})) - \mathcal{E}} - 1}\right)
\ ,
\label{eq:plateau-bounded}
\\
I(\mathcal{E})
& \doteq \frac{\ee^{\mathcal{E}}}{\mathcal{E} \Lambda(\mathcal{E}) }
\int_{-\infty}^{\infty} \! d\Omega \, 
\frac{1- \cos\bigl(\Lambda(\mathcal{E}) \Omega\bigr)}{\Omega^2}  
\, 
\widehat{g} \! \left(\frac{a \Omega}{2\pi}\right) 
\left( \frac{\Omega + \mathcal{E}}{\ee^{\Omega + \mathcal{E}} - 1}\right)
\ , 
\label{eq:plateau-unbounded}
\end{align}
\end{subequations}
and 
\begin{align}
\Lambda(\mathcal{E}) 
\doteq  
a  (2\pi)^{-1} 
\bigl[P\bigl((a \mathcal{E}/(2\pi)\bigr)\tau_p + \tau_s\bigr]
\ . 
\label{eq:Lambda-calE-def}
\end{align}
In \eqref{eq:plateau-bounded} 
we have changed the integration variable by 
$\omega = au/\bigl(2\pi \Lambda(\mathcal{E}) \bigr)$
and in \eqref{eq:plateau-unbounded} by $\omega = a\Omega/(2\pi)$.
Note that 
$\Lambda: \mathbb{R}^+ \rightarrow \mathbb{R}^+$ 
is differentiable, strictly increasing and polynomially bounded,
and $\Lambda (\mathcal{E})\to\infty$ as $\mathcal{E} \to \infty$. 

By \eqref{eq:plateau-prelim-both}
and~\eqref{eq:plateau-both}, 
it suffices to 
show that $H(\mathcal{E})$ is bounded as $\mathcal{E}\to\infty$ and 
$I(\mathcal{E})$ is unbounded as $\mathcal{E}\to\infty$. 

Consider $H(\mathcal{E})$ for $\mathcal{E}>0$. 
Symmetrising the integrand in~\eqref{eq:plateau-bounded}, we find 
\begin{align}
H(\mathcal{E})
& = \frac{1}{2 \mathcal{E}}
\int_{-\infty}^{\infty} 
\! du \, 
\frac{1- \cos u}{u^2} \, 
\widehat{g} \! \left(\frac{au}{2\pi\Lambda(\mathcal{E})}\right) 
\, 
f \bigl(- \mathcal{E}, u/\bigl(\Lambda(\mathcal{E}) \bigr) \bigr) 
\notag
\\
& \hspace{3ex}
+ \frac{1}{1 - \ee^{- \mathcal{E}}}
\int_{-\infty}^{\infty} 
\! du \, 
\frac{1- \cos u}{u^2} \, 
\widehat{g} \! \left(\frac{au}{2\pi\Lambda(\mathcal{E})}\right) 
\ , 
\label{eq:plateau-bounded-1}
\end{align}
where the function $f$ is given by 
\eqref{eq:fs-def} in Appendix~\ref{app:auxiliary-main1}. 
Since the function $(1 - \cos u)/u^2$ is integrable, each term in 
\eqref{eq:plateau-bounded-1} is bounded as $\mathcal{E}\to\infty$: in the 
second term this follows observing that 
$\widehat{g}(\omega)$ is bounded, and in the first term this 
follows using Lemma \ref{Lem1} and observing that
$\omega^2 \widehat{g}(\omega)$ is bounded. 
Hence $H(\mathcal{E})$ is bounded as $\mathcal{E}\to\infty$. 

Consider then $I(\mathcal{E})$ for $\mathcal{E}>0$. 
The strategy is to show that if $I(\mathcal{E})$ 
were bounded as $\mathcal{E}\to\infty$, 
this would imply for $\widehat g$ a stronger falloff than the compact 
support of $g$ allows~\cite{ingham}. 
A~technical subtlety is that 
the factor $\bigl[ 1- \cos\bigl(\Lambda(\mathcal{E}) \Omega\bigr) \bigr]/\Omega^2$ 
in \eqref{eq:plateau-unbounded} 
has zeroes, 
and the contributions to $I(\mathcal{E})$ from neighbourhoods 
of these zeroes will need a careful estimate from below. 

Suppose hence that $I_{\max}$ and $\mathcal{E}_{\text{min}}$ 
are positive constants such that 
$I(\mathcal{E}) \leq I_{\max}$ for all $\mathcal{E} \ge \mathcal{E}_{\text{min}}$. 
We shall show that this leads to a contradiction. 

To begin, let 
$J_\mathcal{E} \doteq \left[-\mathcal{E}-\ee^{-\mathcal{E}/2}, -  \mathcal{E}\right]$. 
Restricting the integral in \eqref{eq:plateau-unbounded} 
to the interval $J_\mathcal{E}$ gives the estimate 
\begin{align}
I(\mathcal{E}) 
& \geq 
\frac{2 \ee^\mathcal{E}}{\mathcal{E} \Lambda(\mathcal{E})} 
\int_{-\mathcal{E}-\ee^{-\mathcal{E}/2}}^{-\mathcal{E}} 
\! d\Omega \, 
\frac{\sin^2 \bigl(\Lambda(\mathcal{E}) \Omega/2\bigr)}{\Omega^2}  
\left( \frac{\Omega + \mathcal{E}}{\ee^{\Omega+\mathcal{E}} - 1}\right) 
\widehat{g} \! \left( \frac{a \Omega}{2 \pi}\right) 
\nonumber 
\\
& \geq 
\frac{2 \ee^{\mathcal{E}/2} Q(\mathcal{E})}{\mathcal{E} \Lambda(\mathcal{E})
\bigl(\mathcal{E} + \ee^{-\mathcal{E}/2}\bigr)^2} 
\inf_{\Omega \in J_\mathcal{E}} \widehat{g} \! 
\left( \frac{a \Omega}{2 \pi} \right)
\ , 
\label{ILowerBd}
\end{align}
where 
\begin{equation}
Q(\mathcal{E}) \doteq 
\inf_{\Omega \in J_\mathcal{E}} \sin^2 \! 
\left( \frac{\Lambda(\mathcal{E}) \Omega}{2} \right)
\ , 
\label{Qdef}
\end{equation}
and we have used the fact that $y/(\ee^y-1) \geq 1$ for $y \leq 0$. 

Next, let 
$T = \bigl\{ \mathcal{E} \in [\mathcal{E}_{\text{min}}, \infty) : 
Q(\mathcal{E}) \geq \ee^{-\mathcal{E}/4} \bigr\}$. 
$T$ is clearly non-empty and unbounded from above. 
For $\mathcal{E} \in T$, 
\eqref{ILowerBd} implies
\begin{equation}
\inf_{\Omega \in J_\mathcal{E}} 
\widehat{g} \! \left( \frac{a \Omega}{2\pi} \right) 
\leq R(\mathcal{E}) \, \ee^{-\mathcal{E}/4}
\ ,
\label{CrudeIneq}
\end{equation}
where
$R(\mathcal{E}) \doteq \tfrac12 
I_{\max} \Lambda(\mathcal{E}) \mathcal{E}(\mathcal{E}+1)^2$. 
For $\mathcal{E} \in T$ 
we hence have
\begin{equation}
\widehat{g}\left(\frac{-a \mathcal{E}}{2 \pi}\right) 
\leq R(\mathcal{E}) \ee^{-\mathcal{E}/4} 
+ C \ee^{-\mathcal{E}/2} \leq \bigl( R(\mathcal{E}) + C \bigr) 
\, \ee^{-\mathcal{E}/4}
\ ,
\end{equation}
where $C \doteq a (2\pi)^{-1} 
\sup_{\mathbb{R}^+} \left| {\widehat{g}\,}' \right|$ is a positive constant. 
Recalling that $\widehat{g}$ is even, we deduce that for any 
$\gamma \in (0, 1/4)$, there exists a positive constant $K$ 
such that 
$\widehat{g} \bigl( a \mathcal{E}/(2\pi) \bigr) 
\leq K \ee^{-\gamma |\mathcal{E}|}$ 
for $|\mathcal{E}| \in T$.

If $T$ comprised all of $[\mathcal{E}_{\text{min}} , \infty \bigr)$, 
the exponential falloff of $\widehat g$ in $T$ would 
now provide a contradiction with the compact support 
of~$g$~\cite{Hormander:1983,Reed-Simon-vol2}.
We shall show that the complement of $T$ in 
$\bigl[\mathcal{E}_{\text{min}} , \infty \bigr)$, 
$T^c = 
\bigl\{ \mathcal{E} \in [\mathcal{E}_{\text{min}}, \infty) : 
Q(\mathcal{E}) < \ee^{-\mathcal{E}/4} \bigr\}$, 
is sufficiently sparse for a contradiction with the 
compact support of $g$ still to ensue by Lemma~\ref{lem:g-S}. 

As a first step, we note that $Q(\mathcal{E})$ has the lower bound 
\begin{equation}
Q(\mathcal{E})  \geq \pi^{-2} \inf_{\Omega \in J_\mathcal{E}} 
\min_{k \in \mathbb{N}_0} |\Lambda(\mathcal{E}) \Omega + 2 \pi k |^2
\ , 
\end{equation}
using \eqref{Qdef}, 
recalling that $\Omega<0$ for $\Omega \in J_\mathcal{E}$, 
and using 
the inequality 
$|\sin(\theta)| \geq 2 |\theta|/\pi$, valid for $|\theta| \leq \pi/2$. 
Since 
$|\Lambda(\mathcal{E})(\Omega+\mathcal{E})| 
\leq \Lambda(\mathcal{E}) \, \ee^{-\mathcal{E}/2}$ 
for 
$\Omega \in J_\mathcal{E}$, it follows using the 
reverse triangle inequality that 
\begin{equation}
\pi Q(\mathcal{E})^{1/2}  \geq \min_{k \in \mathbb{N}_0} 
|\Lambda(\mathcal{E}) \mathcal{E} - 2 \pi k | 
- \Lambda(\mathcal{E}) \, \ee^{-\mathcal{E}/2}
\ .
\label{PreSetTc}
\end{equation}
From \eqref{PreSetTc} and the definition of $T$ 
it then follows that if $\mathcal{E} \in T^c$, 
there exists a $k \in \mathbb{N}_0$ such that 
$|\Lambda(\mathcal{E}) \mathcal{E} - 2 \pi k| < \Lambda(\mathcal{E}) \, 
\ee^{-\mathcal{E}/2} + \pi \ee^{-\mathcal{E}/8}$, and hence further that
for this $k$  
\begin{align}
|\Lambda(\mathcal{E}) \mathcal{E} - 2 \pi k| 
\leq C' \ee^{-\mathcal{E}/8}
\ , 
\label{PreSet-aux-Tc}
\end{align} 
where $C' \doteq \pi 
+ \sup_{\mathcal{E} > 0} \Lambda(\mathcal{E}) \, \ee^{-3\mathcal{E}/8}$
is a positive constant. 

We extend the domain of $\Lambda$ to include the origin by setting 
$\Lambda(0) \doteq \lim_{\mathcal{E}\to0+} \Lambda(\mathcal{E})$. 
Note that $\Lambda(0) >0$ by~\eqref{eq:Lambda-calE-def}. 
Let 
$\Xi: [0,\infty) \to [0,\infty)$ be defined by 
$\Xi(\mathcal{E}) \doteq \mathcal{E}\Lambda(\mathcal{E})$. 
For each $k \in \mathbb{N}_0$, 
let $\mathcal{E}_k$ be the unique solution to 
$\Xi(\mathcal{E}_k) = 2 \pi k$. 
($\mathcal{E}_k$~exists and is unique because 
$\Xi$ is strictly increasing, $\Xi(0) = 0$, 
and 
$\Xi(\mathcal{E}) \to \infty$ as $\mathcal{E} \to \infty$.) 
If $\mathcal{E} \in T^c$, 
we hence see from \eqref{PreSet-aux-Tc} that there 
exists a $k \in \mathbb{N}_0$ such that
$
\left|\Xi(\mathcal{E}) - \Xi(\mathcal{E}_k) \right| \leq C' \ee^{-\mathcal{E}/8}
$. 
Since $\Xi'(v) \geq \Lambda(0)$ for $v\ge0$, this implies 
\begin{equation}
|\mathcal{E} - \mathcal{E}_k| \leq \frac{C'}{\Lambda(0)} \ee^{-\mathcal{E}/8}
\ , 
\label{eq:calE-ineq-4}
\end{equation}
from which we see that 
$\mathcal{E} \geq \mathcal{E}_k - C'/\Lambda(0)$, and using this in 
the exponential in \eqref{eq:calE-ineq-4} gives 
\begin{equation}
|\mathcal{E}- \mathcal{E}_k| 
\leq \frac{C'\ee^{C'/(8 \Lambda(0))}}{\Lambda(0)} \ee^{-\mathcal{E}_k/8}
\ . 
\end{equation}

Collecting, we see that
\begin{equation}\label{eq:Tc}
T^c \subset \bigcup_{k \in \mathbb{N}_0} 
\left\{ \mathcal{E} \in [0, \infty) : 
\, |\mathcal{E} - \mathcal{E}_k | 
\leq \frac{C'\ee^{C'/(8 \Lambda(0))}}{\Lambda(0)} \ee^{-\mathcal{E}_k/8} 
\right\}
\ . 
\end{equation}
This provides for $T^c$ the sparseness that we need.

Let $S$ consist of the union on the right-hand side of \eqref{eq:Tc},
together with its reflection about the origin, and the interval $(-\mathcal{E}_{\text{min}},
\mathcal{E}_{\text{min}})$. Lemma \ref{LemMain2}  
of Appendix~\ref{app:aux-Main2} can now be applied to show that $S$
is a modest set in the sense of Lemma~\ref{lem:g-S}, the parameters
in Lemma~\ref{LemMain2} being chosen so that $\alpha = 1/8$, 
$\beta \in (0, 1/8)$, 
and $C = \delta_0  = 2 C' \ee^{C'/(8 \Lambda(0))}/\Lambda(0)$. 
Since $\Lambda(\mathcal{E}) \leq D(\mathcal{E}+1)^N$ for some 
$D > 0$ and $N > 0$, 
we have $\mathcal{E}_k + 1 \geq (2 \pi k/D)^{1/(N+1)}$, 
and the sequence $\left(\mathcal{E}_k\right)_{k \in \mathbb{N}}$ satisfies the 
condition \eqref{LemSumFinite} in 
Lemma \ref{LemMain2}.
The function $g$ then satisfies the conditions of 
Lemma~\ref{lem:g-S}, by which $g$ must be identically vanishing. 
This contradicts the construction of~$g$. 

Hence the pair of constants $(I_{\max},\mathcal{E}_{\text{min}})$ 
does not exist: 
$I(\mathcal{E})$ is not bounded as $\mathcal{E} \to \infty$. 
\end{proof}

\section{Summary and discussion\label{sec:conclusions}} 

We have asked how long one needs to 
wait for the thermality of the Unruh effect to 
become manifest 
in the response of an 
Unruh-DeWitt particle detector that is coupled linearly 
to a scalar field, 
assuming that the interaction is sufficiently 
weak for linear perturbation theory to be applicable. 

We considered two implementations of the 
long interaction limit: 
an adiabatic scaling, 
which stretches the whole profile of the interaction, 
including the initial switch-on interval and the final switch-off interval, 
and a plateau scaling, which leaves the durations of 
the switch-on and switch-off 
intervals unchanged but 
stretches an intermediate interval during 
which the interaction has constant strength. 

We first showed that the long interaction limit 
with either scaling leads to the 
well-known thermality results, 
in the sense of the detailed balance condition, 
when the detector's energy gap 
$E_{\text{gap}}$ is fixed. However, we also showed that when 
the interaction duration is fixed, detailed balance 
cannot hold in the limit of large~$E_{\text{gap}}$. 
This raised the question of how long one needs to wait for detailed balance 
to hold at a given $E_{\text{gap}}$ when $E_{\text{gap}}$ is large. Our main results 
addressed this question for a massless scalar field in four spacetime 
dimensions, 
in which case the detector's response is identical 
to that of a static detector in a static heat bath. 
We showed that detailed balance at large $E_{\text{gap}}$ can be achieved 
in interaction time that grows as a power-law of $E_{\text{gap}}$ 
with the adiabatic scaling but, 
under mild technical conditions, not with the plateau scaling. 
The upshot is that to achieve detailed balance in power-law interaction time, 
one needs to stretch not just the overall duration 
of the interaction but also the intervals 
in which the interaction is switched on and off. 

Our analysis was motivated in part by experimental considerations: 
a waiting time that grows no faster than a power-law in $E_{\text{gap}}$
would presumably be a physically sensible requirement in 
prospective experimental tests of the Unruh effect. 
However, a deeper motivation was to develop mathematical insight into 
what one might mean by thermality in the detector's response when 
the Wightman function is not invariant under time translations 
along the detector's worldline. 
In such situations the time dependence in the detector's response 
comes not just 
from the switching function that is specified by hand but from 
the genuine time dependence in the quantum field's state or in the detector's motion. 
Examples are a detector in the spacetime of a collapsing star 
during the onset of Hawking radiation~\cite{Juarez-Aubry:2014jba}, 
a detector that falls into a black hole~\cite{Juarez-Aubry:2014jba,Juarez-Aubry:2015dla}, 
and a detector in an expanding cosmology~\cite{Garay:2013dya}. 
Our results show that to characterise the response as approximately thermal 
over some limited interval of time, 
perhaps in a time-dependent local temperature, 
considering the large energy gap limit will not help. 
In particular, a time-dependent temperature defined in an adiabatic regime 
\cite{Kothawala:2009aj,Barcelo:2010pj,Barbado:2011dx,Barbado:2012pt,Barbado:2012fy,Smerlak:2013sga} 
cannot remain valid to arbitrarily high energies when the finite duration of 
the interaction is accounted for.

\section*{Acknowledgments}

We thank Adrian Ottewill and Silke Weinfurtner for helpful comments. 
BAJ-A is supported by Consejo Nacional de Ciencia y Tecnolog\'ia
(CONACYT), Mexico, REF 216072/311506. JL is supported in part by STFC
(Theory Consolidated Grant ST/J000388/1). 


\appendix

\section{Polynomial boundedness of Unruh ${\mathcal{F}}(E)$\label{app:polboundFUnruh}} 

In this appendix we verify that ${\mathcal{F}}(E)$ \eqref{ResponseFn} is 
polynomially bounded for the Unruh effect in the setting 
of Section~\ref{sec:unruh-section}. This follows 
by applying to the 
Unruh effect Wightman function \eqref{eq:rindler-wightman}
the following proposition. 

\begin{prop}
In four spacetime dimensions, 
suppose 
\begin{align}
\mathcal{W}(s) = \lim_{\epsilon\to0+} 
\left( - \frac{1}{4 \pi^2 {(s- \ii \epsilon)}^2} \right) 
+ f(s)
\ , 
\label{eq:Wapp:singminus}
\end{align} 
where the limit indicates the
boundary value in the sense of Definition \ref{def:KMS} 
and $f$ is a smooth function. 
Then $\mathcal{F}(E)$ \eqref{ResponseFn} satisfies 
\begin{align}
\mathcal{F}(E) = 
-\frac{E \Theta(-E)}{2 \pi} \int_{-\infty}^\infty \! d\tau \, [\chi(\tau)]^2 
+ O^\infty(1/E)
\label{eq:F-leadingterm}
\end{align}
as $|E| \to \infty$, 
where $\Theta$ is the Heaviside step function. 
\label{prop:F-leadingterm}
\end{prop}

\begin{proof}
Starting from equation (3.10) in \cite{Louko:2007mu} 
and proceeding as in Section 2 and Appendix A of~\cite{Juarez-Aubry:2014jba}, 
we can write $\mathcal{F}(E)$ \eqref{ResponseFn} as 
\begin{align}
\mathcal{F}(E) & = 
-\frac{E \Theta(-E)}{2 \pi} \int_{-\infty}^\infty \! d\tau \, [\chi(\tau)]^2 
+ \frac{1}{2 \pi^2} \int_0^\infty \! ds \,\frac{\cos(Es)}{s^2} 
\int_{-\infty}^\infty \! d\tau \, \chi(\tau) [\chi(\tau) - \chi(\tau-s)] 
\nonumber 
\\
& \hspace{3ex}
+ 2 \int_{-\infty}^\infty \! d \tau \, \chi(\tau) \int_0^\infty \! ds \, \chi(\tau-s)
\Realpart 
\left[ \ee^{-\ii E s} \left( \mathcal{W}(\tau, \tau-s) + \frac{1}{4 \pi^2 s^2} \right) \right]
\ . 
\label{Louko-Satz2}
\end{align}
In the last term of~\eqref{Louko-Satz2}, we first use the 
stationarity condition 
\eqref{eq:Wstat-def} and \eqref{eq:Wapp:singminus} to replace 
the parentheses
by~$f(s)$, and we then interchange the integrals over 
$s$ and~$\tau$. We note that $\mathcal{W}(-s) = \overline{\mathcal{W}(s)}$ implies 
$f(-s) = \overline{f(s)}$, 
the function 
$s \mapsto \int_{-\infty}^\infty \! d \tau \, \chi(\tau) \chi(\tau-s)$ 
is an even smooth function of compact support, and 
the function 
$s \mapsto s^{-2} \int_{-\infty}^\infty \! d\tau \, \chi(\tau) [\chi(\tau) - \chi(\tau-s)]$ 
is an even smooth function with falloff $O\bigl(s^{-2}\bigr)$ at $s \to \pm\infty$. 
This allows us to write 
\begin{align}
\mathcal{F}(E) & = 
-\frac{E \Theta(-E)}{2 \pi} \int_{-\infty}^\infty \! d\tau \, [\chi(\tau)]^2 
+ \frac{1}{4 \pi^2} \int_{-\infty}^\infty \! ds \,\frac{\cos(Es)}{s^2} 
\int_{-\infty}^\infty \! d\tau \, \chi(\tau) [\chi(\tau) - \chi(\tau-s)] 
\nonumber 
\\
& \hspace{3ex}
+ \int_{-\infty}^\infty \! ds
\, \ee^{-\ii E s} \, f(s) 
\int_{-\infty}^\infty \! d \tau \, \chi(\tau) \chi(\tau-s)
\ . 
\label{Louko-Satz4}
\end{align}
The third term in \eqref{Louko-Satz4} 
is the Fourier 
transform of a smooth compactly supported function, 
and hence $O^\infty(1/E)$ as $|E|\to\infty$~\cite{Hormander:1983,Reed-Simon-vol2}. 
The second term
is the real part of the Fourier 
transform of a smooth function that is a multiple of 
$s^{-2}$ outside a compact interval, 
and may be shown to be 
$O^\infty(1/E)$ as $|E|\to\infty$ by repeated integration by parts~\cite{wong}.
\end{proof}

\section{Auxiliary results for Theorem \ref{thm:Main1}\label{app:auxiliary-main1}} 

\begin{lem}
Let $f: \mathbb{R}^2 \rightarrow \mathbb{R}$ be defined by 
\begin{equation}
f(u,v) \doteq \frac{u + v}{\ee^{u+v} - 1} 
+ \frac{u - v}{\ee^{u - v} - 1} - \frac{2 u}{\ee^{u} - 1}
\ ,
\label{eq:fs-def}
\end{equation}
where the formula is understood in the limiting sense at $u= v$ and $u = -v$. 
Then 
$0 \leq f(u,v) 
\leq v^2/6$.
\label{Lem1}
\end{lem}

\begin{proof}
A direct computation shows that $f$ is even in each of its arguments. 
Since $f(u,0)=0$ for all~$u$, 
it suffices to consider the case 
$u\ge0$ and $v>0$. 
We then have 
$f(0,v) = 2 \bigl[ (v/2)\coth(v/2) -1 \bigr] >0$ and 
$\lim_{u\to\infty} f(u,v) =0$, and it may be shown by an elementary analysis that 
$\partial_u f(u,v)<0$ for $u>0$.  
This implies $0 \leq f(u,v) \leq f(0,v)$. Finally, 
$f(0,v) \leq v^2/6$ may be verified by an elementary analysis of the 
sign of $(x^2+3) \sinh x - 3x \cosh x$. 
\end{proof}

\begin{lem}
The Unruh effect response function 
$\widetilde{\mathcal{F}}_\lambda(E)$ \eqref{eq:tildeFadiab}
and its $\lambda\to\infty$ limit 
$\widetilde{\mathcal{F}}_\infty(E)$ \eqref{eq:tildeFadiab-infty}
satisfy 
$\widetilde{\mathcal{F}}_\infty(E) 
\leq \widetilde{\mathcal{F}}_\lambda(E) 
\leq 
\widetilde{\mathcal{F}}_\infty(E) 
+ (24 \pi a)^{-1}\lambda^{-2} 
|| \omega \widehat{\chi} ||^2$.
\label{Lem2}
\end{lem}

\begin{proof}
Since $\left|\widehat{\chi}(\omega) \right|$ is even in~$\omega$, 
symmetrising the integrands in 
\eqref{eq:tildeFadiab}
and 
\eqref{eq:tildeFadiab-infty}
gives 
\begin{align}
\widetilde{\mathcal{F}}_\lambda(E) 
- 
\widetilde{\mathcal{F}}_\infty(E) 
= \frac{a}{2 {(2\pi)}^3}
\int_{-\infty}^\infty d\omega \left|\widehat{\chi}(\omega) \right|^2 
f \bigl(
2\pi E/a, 2\pi\omega/(\lambda a)
\bigr) 
\ , 
\label{Lem2-proof-int1}
\end{align}
where $f$ is given by~\eqref{eq:fs-def}. 
By Lemma~\ref{Lem1}, 
$ 0 \le 
f \bigl(
2\pi E/a, 2\pi\omega/(\lambda a) \bigr) 
\le \frac{1}{6}
\left(\frac{2\pi}{\lambda a}\right)^2 \omega^2$, 
and inserting this in 
\eqref{Lem2-proof-int1} completes the proof. 
\end{proof}

\begin{lem}
Let $h: \mathbb{R}^{+}\times\mathbb{R}^{+} \to \BbbR$ be defined by 
$h(u,v) = (u - v)/(\cosh u-\cosh v)$, 
where the formula is understood in the limiting sense at $u=v$. 
$h$~is strictly positive, and it is strictly decreasing in each of its arguments.
\label{LemBound1}
\end{lem}
\begin{proof}
It is immediate that $h$ is strictly positive, and $h(u,v) = h(v,u)$. 
An elementary analysis shows that 
$\partial_v h(u,v) < 0$. 
\end{proof}

\begin{lem}
Under the assumptions of Theorem~\ref{thm:Main1}, 
with $\lambda(E) = (2\pi E/a)^{1+p}$ and 
$\mathcal{G}_\text{\rm{est}}(E)$ given by~\eqref{eq:G-est-def}, 
the inequality \eqref{eq:G-est-used}
holds for sufficiently large~$E$. 
\label{Lem-n1}
\end{lem}

\begin{proof}
Let $E>0$ and let 
\begin{equation}
\mathcal{G}(E) \doteq \frac{\exp(2 \pi E/a)}{2\pi E/a}
\left(\widetilde{\mathcal{F}}_{\lambda(E)}(E) - \widetilde{\mathcal{F}}_{\infty}(E) \right)
\ . 
\label{GMain1}
\end{equation}
We need to show that $\mathcal{G}(E) \le \mathcal{G}_\text{est}(E)$ 
for sufficiently large~$E$. 

Let $g: \mathbb{R}^2 \rightarrow \mathbb{R}$ be defined by 
\begin{align}
g(u,v) & \doteq \ee^u f(u,v) 
\notag
\\[1ex]
&
=
\frac{u + v}{\ee^{v} - \ee^{-u}} 
+ \frac{u - v}{\ee^{- v} - \ee^{-u}} 
- \frac{2 u}{1 - \ee^{-u}}
\ ,
\label{app5:g}
\end{align}
where $f$ was defined in~\eqref{eq:fs-def}, 
and the last formula in \eqref{app5:g} 
is understood in the limiting sense at $u= v$ and $u = -v$. 
Symmetrising the integrands in \eqref{eq:tildeFadiab}
and 
\eqref{eq:tildeFadiab-infty} as in the proof of Lemma~\ref{Lem2}, 
we may decompose \eqref{GMain1} as 
\begin{align}
\mathcal{G}(E) = 
\mathcal{G}_1
+ \mathcal{G}_2
+ \mathcal{G}_3
\ , 
\label{5:Gs-sum}
\end{align}
where 
\begin{subequations}
\label{5:Gs}
\begin{align}
\mathcal{G}_1
& \doteq 
\frac{a}{(2 \pi)^4 (E/a)}  
\int_0^\kappa \! d\omega \, \left| \widehat{\chi}(\omega) \right|^2 
g\bigl(2\pi E/a, 2 \pi \omega/(\lambda a)\bigr)
\ , 
\label{5:Gs-1}
\\
\mathcal{G}_2
& \doteq 
\frac{a}{(2 \pi)^4 (E/a)}  
\int_\kappa^{E \lambda} \! d\omega \, \left| \widehat{\chi}(\omega) 
\right|^2 g\bigl(2\pi E/a, 2 \pi \omega/(\lambda a)\bigr)
\ , 
\label{5:Gs-2}
\\
\mathcal{G}_3
& \doteq \frac{a}{(2 \pi)^4 (E/a)}  
\int_{E \lambda}^\infty \! d\omega \, \left| \widehat{\chi}(\omega) \right|^2 
g\bigl(2\pi E/a, 2 \pi \omega/(\lambda a)\bigr)
\ , 
\label{5:Gs-3}
\end{align}
\end{subequations}
and we have denoted $\lambda(E) = (2\pi E/a)^{1+p}$ 
by just~$\lambda$. 
We may assume $E$ to be so large that $E\lambda > \kappa$. 
From Lemma \ref{Lem1} we then see that 
$\mathcal{G}_1$, 
$\mathcal{G}_2$
and
$\mathcal{G}_3$
are strictly positive. We need to bound $\mathcal{G}_1$, 
$\mathcal{G}_2$
and
$\mathcal{G}_3$ 
from above. 

\subsection*{Bounding $\mathcal{G}_1$}

Consider $\mathcal{G}_1$~\eqref{5:Gs-1}. 
We write $u = 2\pi E/a$ and $v = 2\pi \omega/a$, 
and note that $v\le 2\pi\kappa/a$ and $u > v/u^{1+p}$ since by assumption 
$0\le\omega\le\kappa$ and $E\lambda > \kappa$.  
In the integrand in~\eqref{5:Gs-1}, we then have, using~\eqref{app5:g}, 
\begin{align}
g\bigl(u,v/u^{1+p}\bigr) 
& \leq \frac{u + v/u^{1+p}}{1-\ee^{-u}}  
+\frac{u - v/u^{1+p}}{\ee^{-v/u^{1+p}}
-\ee^{-u}} - \frac{2u}{1-\ee^{-u}} 
\notag
\\
& = \left( u - \frac{v}{u^{1+p}}\right) 
\left(\frac{1}{\ee^{-v/u^{1+p}}-\ee^{-u}} -\frac{1}{1-\ee^{-u}}\right) 
\notag
\\
& \leq u 
\left(\frac{1}{\ee^{-(2 \pi \kappa/a)/ u^{1+p}}-\ee^{-u}} 
-\frac{1}{1-\ee^{-u}}\right)
\notag
\\
& = u \left( 
\frac{2\pi\kappa/a}{u^{1+p}} 
+ O \! \left( u^{-2(1+p)}\right) 
\right)
\notag
\\
& =
\frac{2\pi\kappa/a}{u^{p}}
+ O \! \left( u^{-(1+2p)}\right) 
\ . 
\end{align}
Using this in \eqref{5:Gs-1}, 
and using the evenness of $\left| \widehat{\chi}(\omega) \right|$, 
we find 
\begin{equation}
\mathcal{G}_1 \leq 
\frac{\kappa}{2{(2 \pi)}^2} 
\left[ {(2\pi E/a)}^{-(p+1)} + O \! \left( {(E/a )}^{-2(1+p)} \right) \right]
|| \widehat{\chi} ||^2
\ .
\label{boundpiece1}
\end{equation}

\subsection*{Bounding $\mathcal{G}_2$}

Consider $\mathcal{G}_2$~\eqref{5:Gs-2}. 
Writing again $u = 2\pi E/a$ and $v = 2\pi \omega/a$, 
we now have $2 \pi \kappa/a \le v \le u^{2+p}$. 
In the integrand in~\eqref{5:Gs-2}, we have 
\begin{align}
g\bigl(u,v/(u^{1+p})\bigr) 
& =  
\left( u - \frac{v}{u^{1+p}}\right) 
\left(\frac{1}{\ee^{-v/u^{1+p}}-\ee^{- u}} -\frac{1}{\ee^{v/u^{1+p}} 
-\ee^{- u}}\right) \nonumber \\
& \hspace{2ex}
+ 2u 
\left(\frac{1}{\ee^{v/u^{1+p}} -  \ee^{-u}} - \frac{1}{1- \ee^{-u}} \right)
\notag
\\
& \le 
\left( u - \frac{v}{u^{1+p}}\right) 
\frac{\ee^u \sinh \left(v/u^{1+p}\right)}{\cosh u-\cosh \bigl(v/u^{1+p}\bigr)}
\notag
\\
& \le 
\frac{ u \ee^u \sinh \bigl(v/u^{1+p}\bigr)}{\cosh u-1}
\ , 
\label{gbound}
\end{align}
where the first inequality comes by observing that 
the last term preceding the ``$\le$'' sign 
is negative, and the second inequality follows from 
Lemma~\ref{LemBound1}. From \eqref{5:Gs-2} we hence have 
\begin{align}
\mathcal{G}_2
& \le   \frac{a \, \ee^{2 \pi E/a}}{{(2 \pi)}^3 \left[\cosh(2 \pi E/a)-1 \right]} 
\int_\kappa^{E \lambda} \! d\omega \, \left| \widehat{\chi}(\omega) \right|^2 
\sinh \bigl( 2 \pi \omega/(\lambda a) \bigr)
\ . 
\label{Gest1}
\end{align}
Note that the factor in front of the integral in \eqref{Gest1} is 
$(2a) {(2\pi)}^{-3} \left[1 + O\bigl(\ee^{-2\pi E/a}\bigr) \right]$. 

Let $\omega_0 \doteq (a/2\pi) (2\pi E/a)^{(1 + p + q^{-1})/2 }$. 
Since $p > q^{-1} -1$ by assumption, it follows that 
$\kappa < \omega_0 < E \lambda$
for sufficiently large $E$, and we now assume $E$ to be this large. 
We shall bound separately the contributions to \eqref{Gest1} from 
$\kappa \le \omega \le \omega_0$ and from $\omega_0 \le \omega \le E\lambda$. 

In the contribution from 
$\kappa \le \omega \le \omega_0$, we have 
$\sinh\bigl(2\pi\omega/(\lambda a)\bigr) \le \sinh\bigl(2\pi\omega_0/(\lambda a)\bigr)$, 
and hence 
\begin{align}
\int_\kappa^{\omega_0} \! d\omega \, 
\left| \widehat{\chi}(\omega) \right|^2 
\sinh \bigl( 2 \pi \omega/(\lambda a) \bigr)
& \leq 
\tfrac12 
|| \widehat{\chi} ||^2 
\sinh \bigl( 2 \pi \omega_0/(\lambda a) \bigr)
\nonumber 
\\ 
& = \tfrac12 || \widehat{\chi}||^2 
\sinh \! 
\left( 
(2 \pi E/a)^{-(1+p-q^{-1})/2} 
\right) 
\nonumber 
\\
& = 
\tfrac12 
||\widehat{\chi}||^2 
(2\pi E/a)^{-(1+p-q^{-1})/2} 
\notag
\\
& \hspace{5ex}
\times 
\left[
1 + O \! \left((E/a)^{-(1+p-q^{-1})} \right) 
\right]
\ , 
\end{align} 
using the evenness of $\left| \widehat{\chi}(\omega) \right|$. 

In the contribution from 
$\omega_0 \le \omega \le E\lambda$, 
we use 
$\sinh\bigl(2\pi\omega/(\lambda a)\bigr) \le 
\tfrac12 \exp \bigl(2\pi\omega/(\lambda a)\bigr)$ 
and the strong Fourier decay property~\eqref{eq:Main1-inputest}, obtaining 
\begin{align}
\left| \widehat{\chi}(\omega) \right|^2 \sinh \left( \frac{2 \pi \omega}{\lambda a} \right)  
& \leq 
\frac{C^2}{2 \kappa^{2}} 
\left(B+ \frac{\omega}{\kappa} \right)^{2r} 
\exp \! \left[
-2A \left(\frac{\omega}{\kappa}\right)^q + \frac{2 \pi \omega}{a \lambda} 
\right] 
\nonumber 
\\
& \leq 
\frac{C^2}{2\kappa^{2}} 
\left(B+ \frac{\omega}{\kappa} \right)^{2r} 
\exp \! 
\left[ -2A \left(\frac{\omega_0}{\kappa}\right)^q + \frac{2 \pi E}{a} \right]
\nonumber 
\\
& = \frac{C^2}{2\kappa^{2}} 
\left(B+ \frac{\omega}{\kappa} \right)^{2r} 
\exp \! 
\left[ 
-2A \left(\frac{a}{2 \pi \kappa}\right)^q 
\left(\frac{2 \pi E}{a} \right)^{\left(1 + q( 1 + p) \right)/2} + \frac{2 \pi E}{a} 
\right]
\ .
\label{Gest1-a3}
\end{align}
Since $q(1 + p) > 1$ by assumption, the exponential factor in 
\eqref{Gest1-a3} falls off as $E\to\infty$ faster than any power of $E$, while the integral 
$\int_{\omega_0}^{E\lambda} d\omega \, \left[B + (\omega/\kappa) \right]^{2r}$ 
has only power-law growth in $E$ as $E\to\infty$. Hence the contribution to \eqref{Gest1}
from 
$\omega_0 \le \omega \le E\lambda$ falls off faster than any power as $E\to\infty$. 

Collecting these estimates, we have 
\begin{equation}
\mathcal{G}_2
\le  
\frac{a ||\widehat{\chi}||^2 }{{(2 \pi)}^3} 
\left( \frac{2\pi E}{a}\right)^{-(1+p-q^{-1})/2} 
\left[
1 + O \! \left((E/a)^{-(1+p-q^{-1})} \right) 
\right]
\ .
\label{boundpiece2}
\end{equation}

\subsection*{Bounding $\mathcal{G}_3$}

Consider $\mathcal{G}_3$~\eqref{5:Gs-3}. 
Writing again $u = 2\pi E/a$ and $v = 2\pi \omega/a$, 
we now have $u^{2+p} \le v$.
Using the penultimate expression in \eqref{gbound} 
to bound the integrand in~\eqref{5:Gs-3}, we obtain 
\begin{equation}
\mathcal{G}_{(E \lambda, \infty)} 
\leq 
\frac{\ee^{2 \pi E/a}}{(2 \pi)^3 (E/a)} \int_{E \lambda}^\infty \! d \omega \, 
\left|\widehat{\chi}(\omega)\right|^2 
\left( E - \frac{\omega}{\lambda}\right) 
\frac{\sinh \bigl(2 \pi \omega /(a\lambda)\bigr)}{\cosh (2 \pi E/a) 
-\cosh \bigl(2 \pi \omega/(a \lambda) \bigr)}
\ .
\label{app5:G3}
\end{equation}

Since $q(p+1)>1$ by assumption, we have $1 + p  < q(2 +p) + p$, and we may choose 
$s$ such that $1 + p < s < q(2 +p) + p$. Let $\omega_1 \doteq E  (2 \pi E/a)^s$. 
Assuming $2 \pi E/a > 1$, we then have $E\lambda < \omega_1$. 
In~\eqref{app5:G3}, we denote the contributions from 
$E\lambda \le \omega \le \omega_1$ and $\omega_1 \le \omega < \infty$ by 
respectively $\mathcal{H}_1$ and $\mathcal{H}_2$. 
We shall show that both $\mathcal{H}_1$ and $\mathcal{H}_2$
fall off faster than any power of $E$ as $E\to\infty$. 

For $\mathcal{H}_1$, Lemma \ref{LemBound1} gives 
\begin{align}
\mathcal{H}_1 
\leq 
\frac{a \ee^{2\pi E/a}}{(2 \pi)^4 (E/a) \sinh(2 \pi E/a)} 
\int_{E \lambda}^{\omega_1} \! d \omega \, 
\left|\widehat{\chi}(\omega)\right|^2 
\sinh \bigl(2 \pi \omega /(a\lambda)\bigr)
\ , 
\label{eq:H1}
\end{align}
and for the integrand in \eqref{eq:H1} we may proceed as in \eqref{Gest1-a3}
to obtain
\begin{align}
\left|\widehat{\chi}(\omega)\right|^2 
& \sinh \left(\frac{2 \pi \omega}{a \lambda} \right) 
 \leq 
\frac{C^2}{2 \kappa^{2}} 
\left(B+ \frac{\omega}{\kappa} \right)^{2r} 
\exp \! \left[
-2A \left(\frac{\omega}{\kappa}\right)^q + \frac{2 \pi \omega}{a \lambda} 
\right]
\nonumber 
\\
& \leq 
\frac{C^2}{2 \kappa^{2}} 
\left(B+ \frac{\omega}{\kappa} \right)^{2r} 
\exp \! \left[
-2A \left(\frac{E\lambda}{\kappa}\right)^q + \frac{2 \pi \omega_1}{a \lambda} 
\right]
\nonumber 
\\
& = \frac{C^2}{2\kappa^{2}} \left(B+ \frac{\omega}{\kappa} \right)^{2r} 
\exp \! 
\left[ 
-2A \left(\frac{a}{2 \pi \kappa}\right)^q 
\left(\frac{2 \pi E}{a} \right)^{q(2+p)} 
+ \left(\frac{2 \pi E}{a} \right)^{s-p} 
\right]
\ .
\label{supfastG31}
\end{align}
Since $1 < s-p < q(2+p)$, the exponential factor in 
\eqref{supfastG31} shows that $\mathcal{H}_1$ falls off 
faster than any power of $E$ as $E\to\infty$. 

For $\mathcal{H}_2$, we have 
\begin{align}
\frac{( E - \omega/\lambda) 
\sinh \bigl(2 \pi \omega /(a\lambda)\bigr)}{\cosh (2 \pi E/a) -\cosh \bigl(2 \pi \omega/(a \lambda) \bigr)} 
& =  
\left( \frac{ \omega}{\lambda}
- E \right) 
\tanh\left(\frac{2\pi \omega}{a\lambda}\right)
\left(1-\cfrac{\cosh(2 \pi E/a)}{\cosh\bigl(2 \pi \omega /(a\lambda)\bigr)}\right)^{-1} 
\nonumber 
\\
& \leq  
\frac{ \omega}{\lambda} 
\left(1-\cfrac{\cosh(2 \pi E/a)}{\cosh\bigl(2 \pi \omega_1 /(a\lambda)\bigr)}\right)^{-1} 
\nonumber 
\\
& = 
\frac{ \omega}{(2\pi E/a)^{1+p}} 
\left(1-\cfrac{\cosh(2 \pi E/a)}{\cosh\bigl({(2 \pi E/a)}^{s-p}\bigr)}\right)^{-1}
\nonumber 
\\
& \le 
\frac{2\omega}{(2\pi E/a)^{1+p}} 
\ ,
\label{boundshyp}
\end{align}
where the last inequality holds for sufficiently large $E$ because $s-p>1$. 
For $E$ this large, we hence have 
\begin{equation}
\mathcal{H}_2 
\leq 
\frac{2 \ee^{2 \pi E/a}}{(2 \pi)^2 (2\pi E/a)^{2+p}} 
\int_{\omega_1}^\infty \! d \omega \, \omega 
\left|\widehat{\chi}(\omega)\right|^2 
\ .
\label{app5:H2}
\end{equation}

Since $q(s+1) > 1$, for sufficiently large $E$ we have $A(\omega_1/\kappa)^q \ge 2 \pi E/a$. 
For $E$ this large, the strong Fourier decay property \eqref{eq:Main1-inputest}
gives for the integrand in \eqref{app5:H2} the estimate 
\begin{align}
\omega 
\left|\widehat{\chi}(\omega)\right|^2 
& \le 
\frac{C^2}{\kappa^{2}} \left(B+ \frac{\omega}{\kappa} \right)^{2r} 
\omega 
\exp \! 
\left[ 
-2A \left(\frac{\omega}{\kappa}\right)^q \right]
\notag
\\
& \le 
\frac{C^2}{\kappa^{2}} \left(B+ \frac{\omega}{\kappa} \right)^{2r} 
\omega 
\exp \! 
\left[ 
-A \left(\frac{\omega}{\kappa}\right)^q - A \left(\frac{\omega_1}{\kappa}\right)^q 
\right]
\notag
\\
& \le 
\ee^{-2 \pi E/a}
\frac{C^2}{\kappa^{2}} \left(B+ \frac{\omega}{\kappa} \right)^{2r} 
\omega 
\exp \! 
\left[ 
-A \left(\frac{\omega}{\kappa}\right)^q
\right]
\ . 
\label{boundpiece3}
\end{align}
From 
\eqref{app5:H2}
and 
\eqref{boundpiece3}
we see that 
$\mathcal{H}_2$ falls off faster than any power of $E$ as $E\to\infty$. 

\subsection*{Conclusion of the proof} 

Collecting the power-law estimates \eqref{boundpiece1}
and \eqref{boundpiece2} for $\mathcal{G}_1$ and $\mathcal{G}_2$, 
and the faster than power-law falloff of~$\mathcal{G}_3$, 
we see that the weakest estimate for $\mathcal{G}(E)$ 
is~\eqref{boundpiece2}. 
Hence 
\begin{equation}
\mathcal{G}(E) \le 
\frac{2a ||\widehat{\chi}||^2}{{(2 \pi)}^3} 
\left( \frac{2\pi E}{a}\right)^{-(1+p-q^{-1})/2} 
\label{eq:boundfinal-G}
\end{equation}
for sufficiently large~$E$. 

\end{proof}

\section{Auxiliary result for Theorem \ref{thm:Main2}\label{app:aux-Main2}} 

\begin{lem}
Let $\left(\epsilon_k\right)_{k \in \mathbb{N}}$ 
be a strictly positive sequence such that
\begin{equation}
\sum_{k = 1}^\infty \ee^{-\beta \epsilon_k} < \infty
\label{LemSumFinite}
\end{equation}
with $\beta >0$ and let $\delta_k = C \, \ee^{-\alpha \epsilon_k}$ 
for some $C>0$ and $\alpha > \beta$. 
Define $\epsilon_{-k} = -\epsilon_k$ and $\delta_k = \delta_{-k}$ 
for $k \in \mathbb{N}$, 
and let 
\begin{align}
S = \bigcup_{k \in \mathbb{Z}\setminus 0} 
\left\{\epsilon: \left|\epsilon - \epsilon_k \right| < \delta_k \right\} 
\cup (-\delta_0, \delta_0).
\end{align} 
If $F: \mathbb{R} \rightarrow \mathbb{C}$ is a locally 
integrable polynomially bounded function with
$\supp(F) \subset S$ then the inverse Fourier transform of~$F$, 
\begin{equation}
\mathcal{F}^{-1}[F](z) = \frac{1}{2 \pi} \int_{-\infty}^\infty \! d\epsilon \, 
\ee^{\ii z \epsilon} F(\epsilon)
\ ,
\label{LemFourierExists}
\end{equation}
is analytic in the strip $|\Imagpart(z)| < \alpha - \beta$. 
In particular, $S$ is a modest 
set in the sense introduced in Section~\ref{sec:prelim}.
\label{LemMain2}
\end{lem}

\begin{proof}
Since $F$ is polynomially bounded, there exist $D > 0$ and $N>0$ such that 
$|F(\epsilon)| \leq D (1 + |\epsilon|^N)$. We therefore have the estimates
\begin{align}
2 \pi \left| \mathcal{F}^{-1}[F](z) \right| 
& \leq 
\int_{-\infty}^\infty \! d\epsilon \, 
\left| \ee^{\ii \epsilon z} F(\epsilon)\right| 
\leq 2 D \int_{0}^\infty \! d\epsilon \, 
\ee^{|\Imagpart(z)| |\epsilon|} \left(1+|\epsilon|^N \right) 
\nonumber 
\\
& \leq 
2 D \sum_{k = 1}^\infty (2 \delta_k ) \ee^{|\Imagpart(z)| 
(\epsilon_k + \delta_k)} \left(1+(\epsilon_k + \delta_k)^N \right) 
\nonumber 
\\
& \leq 
4 C D \sum_{k = 1}^\infty \left(1+(\epsilon_k + \delta_k)^N \right) 
\ee^{|\Imagpart(z)| (\epsilon_k + \delta_k)-\alpha \epsilon_k}.
\label{boundFourier}
\end{align}

From \eqref{LemSumFinite} it follows that 
$\epsilon_k \to \infty$ as $k\to\infty$, and hence 
$\delta_k \to 0$ as $k\to\infty$. 
If $|\Imagpart(z)| < \alpha - \beta$, 
for sufficiently large $k$ we may hence estimate the 
terms in \eqref{boundFourier} by 
\begin{align}
& \left(1 + (\epsilon_k+\delta_k)^N\right)
\ee^{|\Imagpart(z)|(\epsilon_k+\delta_k) - \alpha\epsilon_k}
\notag
\\
& \ \ \le \left(1 + (\epsilon_k+1)^N\right)
\ee^{|\Imagpart(z)|(\epsilon_k+1) - \alpha\epsilon_k}
\notag
\\
& \ \ = \ee^{|\Imagpart(z)|} \left(1 + (\epsilon_k+1)^N\right)
\ee^{ - (\alpha - \beta - |\Imagpart(z)|)\epsilon_k} 
\, 
\ee^{- \beta\epsilon_k}
\notag
\\
& \ \ \le 
\ee^{\alpha - \beta} \, 
\ee^{- \beta\epsilon_k}
\  . 
\end{align}
This shows that $\mathcal{F}^{-1}[F](z)$ exists in the strip 
$|\Imagpart(z)| < \alpha - \beta$. 

To show that $\mathcal{F}^{-1}[F](z)$ is analytic in the strip 
$|\Imagpart(z)| < \alpha - \beta$, we may use the inequality 
\begin{align}
\left|\frac{\ee^{\ii (z+h)\epsilon} 
- \ee^{\ii z\epsilon}}{h} 
- \ii \epsilon \, \ee^{\ii z\epsilon}\right| 
\le
\tfrac12 
|\epsilon| \ee^{(|\Imagpart(z)| + |h| )\epsilon}
\ , 
\end{align}
valid for $h\in\mathbb{C}\setminus\{0\}$, 
together with estimates similar to those above, to provide 
a dominated convergence argument that justifies differentiating 
\eqref{LemFourierExists} under the integral sign, with the outcome 
\begin{align}
\frac{d}{dz} \mathcal{F}^{-1}[F](z) 
= 
\frac{\ii}{2\pi} \int_{-\infty}^\infty d\epsilon \, 
\epsilon \, \ee^{\ii z \epsilon} F(\epsilon)
\ . 
\end{align}
\end{proof}

\end{document}